\documentclass[a4paper,11pt]{article}

\usepackage{amsthm,amsmath,amssymb,amsfonts,pstricks,pst-node,graphics,graphicx}

\usepackage{amsmath,amsfonts,amsthm,relsize}
\usepackage{amsmath,amsfonts,amsthm,amssymb}
\usepackage{graphicx,pstricks}
\newtheorem{theorem}{Theorem}
\numberwithin{theorem}{section}

\newtheorem{example}[theorem]{Example}

\newtheorem{proposition}[theorem]{Proposition}

\theoremstyle{definition}
\newtheorem{definition}[theorem]{Definition}

\newtheorem{remark}[theorem]{Remark}

\def\RP{\mathbb {RP}}
\def\R{\mathbb R}
\def\Z{\mathbb{Z}}

\def\so{\mathfrak{so}}
\def\g{\mathfrak{g}}

\def\HH{\mathcal{H}}
\def\T{\mathcal{T}}
\def\A{\mathcal{A}}

\def\SL{\mathrm{SL}}
\def\GL{\mathrm{GL}}
\def\SO{\mathrm{SO}}
\def\M{\mathcal{M}}

\def\P{\mathcal{P}}

\def\s{\varsigma}

\def\sinr{\sin{\norm{x_r}}}
\def\sinro{\sin{\norm{x_{r+1}}}}
\def\cosr{\cos{\norm{x_r}}}
\def\cosro{\cos{\norm{x_{r+1}}}}

\def\r{{\bf r}}
\def\k{{\bf k}}

\def\h{{\mathfrak h}}

\def\m{{\mathfrak m}}
\def\v{{\bf v}}

\def\g{\mathfrak{g}}

\newcommand\deter[1]{\left|#1\right|}

\def\M{\mathbf{M}}
\newcommand\norm[1]{|\!|#1|\!|}

\begin{document}
\title{Discrete moving frames  and discrete integrable systems}
\author{ Elizabeth Mansfield, Gloria Mar\'i Beffa \& Jing Ping Wang\footnote{This paper is supported by ELM's EPSRC grant EP/H024018/1, GMB's NSF grant DMS \#0804541 and JPW's EPSRC grant EP/I038659/1.}}
\date{}
\maketitle
\begin{abstract} 
Group based moving frames have a wide range of applications, from the classical equivalence problems in differential geometry to more
modern applications such as computer vision.  Here we describe what we call a 
\textit{discrete group based moving frame}, which is essentially a sequence of moving frames with overlapping domains. 
{ We demonstrate a small set of generators of the algebra of invariants, which we call the discrete Maurer--Cartan invariants, for which
there are recursion formulae. We show that this offers significant computational advantages over a single moving frame for our study 
 of discrete integrable systems.}
We demonstrate that the discrete analogues of some curvature flows  lead naturally to Hamiltonian
pairs, which generate integrable differential-difference systems. In particular, 
we show that in the centro-affine plane and the projective space, the Hamiltonian pairs obtained can be transformed into the known
Hamiltonian pairs for the Toda and modified Volterra lattices respectively under  Miura transformations. { We also show that a specified invariant map of polygons in the centro-affine plane
can be transformed to the integrable discretization of the Toda Lattice}. Moreover, we describe in detail the case of 
{discrete flows in} the homogeneous $2$-sphere and we obtain realizations of equations of Volterra type as evolutions of polygons on the sphere.
 \end{abstract}
 
\centerline{\it Dedicated to Peter Olver in celebration of his 60th birthday}
\section{Introduction}
The notion of a moving frame is associated with \'Elie Cartan \cite{Cartan},
who used it to solve equivalence
problems in differential geometry.
Moving frames were further developed and
applied in a substantial body of work, in particular
to differential geometry and (exterior) differential systems,
see for example papers by  Green \cite{Green}
and Griffiths \cite{Griffiths}.
From the point of view of symbolic computation, a breakthrough
in the understanding of Cartan's methods came in a series of papers
by Fels and Olver \cite{refFelOla,refFelOlb},
 Olver \cite{OJS,Ogen}, Hubert \cite{hubertAA,hubertAC,hubertAD}, and Hubert and Kogan \cite{hubertA,hubertB}, 
which provide a coherent, rigorous and constructive moving
frame method free from any particular application, and hence
applicable to a huge
range of examples,  from classical invariant theory 
to numerical schemes.

For the study of differential invariants, one of the
main results of the Fels and Olver papers is the derivation of symbolic formulae
for differential invariants and their invariant differentiation. The book 
\cite{mansfield} contains a detailed exposition of the calculations for the 
resulting symbolic invariant calculus. Applications include the integration of 
Lie group invariant differential equations,  to the Calculus of Variations
and Noether's Theorem,
(see also \cite{kogan, GonMan}), and to integrable systems (\cite{ManKamp, M1, M2, M3}).

The first results for the computation of discrete  invariants using group-based moving frames were given by Olver \cite{OJS} who calls them
joint invariants;
modern applications to date include computer vision  \cite{Ogen} and numerical schemes
for systems with a Lie symmetry \cite{chhay, kimA, kimB, kimC, ManHy}.
While moving frames for discrete applications as formulated
by Olver do give generating sets of discrete invariants, the recursion formulae for differential invariants
 that were so successful for the application
of moving frames to  calculus based applications do not generalize well to these discrete invariants. In particular,
  these generators do not seem to have recursion formulae under the shift operator that are computationally useful.
  To overcome this computational problem, we introduce a \textit{discrete moving frame} which is essentially 
a sequence of frames\footnote {A sequence of moving frames was also used in \cite{kimB} to minimize the accumulation of errors in an invariant numerical method.}, and prove discrete recursion formulae for a small computable generating sets of invariants, which we call the \textit{discrete Maurer--Cartan invariants}.

{ We show that our definitions and constructions arise naturally and are useful for the study of discrete integrable systems. These arise as analogues of curvature flows for polygon evolutions in homogeneous spaces, and this is the focus of the second half of the paper. }
The study of discrete integrable systems is rather new. It began with discretising continuous integrable systems in 1970s.
 The most well known discretization of the Korteweg-de Vries equation (KdV) is 
the Toda lattice~\cite{toda}
\begin{eqnarray}\label{todau}
\frac{{\rm d}^2 u_s}{{\rm d} t^2}=\exp(u_{s-1}-u_s)-\exp(u_s-u_{s+1}) .
\end{eqnarray}
Here the dependent variable $u$ is a function of time $t$ and discrete variable $s\in \mathbb{Z}$.
We can obtain a finite-dimensional version by picking $N\in\mathbb{N}$ and restricting to $1\leq s \leq N$ 
subject to one of two types of boundary conditions: open-end ($u_0=u_{N}=0$) or periodic 
($u_{s+N} = u_s$ for all $s$ and some period $N$).
Using the Flaschka \cite{Flaschka1,Flaschka2} coordinates
\begin{eqnarray*}
 q_s=\frac{{\rm d} u_s}{{\rm d} t}, \qquad p_s=\exp(u_{s}-u_{s+1}),
\end{eqnarray*}
we rewrite the Toda lattice (\ref{todau}) in the form
\begin{eqnarray}\label{todapq}
\frac{{\rm d} p_s}{{\rm d} t}=p_s (q_s -q_{s+1}), \qquad \frac{{\rm d} q_s}{{\rm d} t}=p_{s-1}-p_s .
\end{eqnarray}
Its complete integrability was first established by Flaschka and Manakov \cite{Flaschka1,Flaschka2,manakov}.
They constructed the Lax representation of system (\ref{todapq}) and further solved it by the
inverse scattering method. 

Another famous integrable discretization of the KdV equation is the Volterra lattice~\cite{manakov,kacm}
\begin{eqnarray*}
\frac{{\rm d} q_s}{{\rm d} t} =q_s (q_{s+1}-q_{s-1}).
\end{eqnarray*}
By the Miura transformation $q_s=p_s p_{s-1}$, it is related to the equation
\begin{eqnarray}\label{volt}
\frac{{\rm d} p_s}{{\rm d} t} =p_s^2 (p_{s+1}-p_{s-1}),
\end{eqnarray}
which is the modified Volterra lattice, an integrable discretization of the modified KdV equation. 

Since the establishment of their integrability,  a great deal of work has been contributed to the study of their other
integrable properties including  Hamiltonian structures, higher symmetry flows and 
 $r$-matrix structures, as well as to the establishment of integrability for other systems 
and further discretising differential-difference integrable systems to obtain integrable maps. 
{ For example, the time discretisation of the Toda lattice (\ref{todapq}) leads to the integrable map
$(p,q)\mapsto (\tilde p, \tilde q)$ defined by 
\begin{eqnarray}\label{todamap}
\tilde p_s=p_s \frac{\beta_{s+1}}{\beta_s}, \qquad 
\tilde q_s=q_s+c \left(\frac{p_s}{\beta_s}-\frac{p_{s-1}}{\beta_{s-1}} \right),
\end{eqnarray}
where $c\in \R$ is constant and 
the function $\beta$ is given by the recurrent relation\\ $$\beta_s=1+c q_s -c^2 \frac{p_{s-1}}{\beta_{s-1}}.$$}
Some historical background about the development of the theory of discrete integrable systems 
can be found in \cite{suris03}. Some classification results for such integrable systems including 
the Toda and Volterra lattice were obtained by the symmetry approach \cite{Yami}.

In this paper, {we introduce the concept of discrete moving frames, and under conditions which are satisfied for the range of examples we study, we prove  theorems analogous to the classical results of the continuous case: generating properties of Maurer--Cartan invariants, a replacement rule, recursion formulas, and general formulas for invariant evolutions of polygons. Once the ground work is in place, we study the evolution induced on the Maurer--Cartan invariants by invariant evolutions of $N$-gons, the so-called {\it invariantizations}. 
We consider the resulting equations to be defined on infinite lattices, i.e.
$s\in \mathbb{Z}$ for both $N$-periodic and non-periodic cases. 
We will show that the invariantization of certain} time evolutions of $N$-gons (or so-called {\it twisted $N$-gons} in the periodic case) in the centro-affine plane and the projective 
line $\RP^1$ naturally lead to Hamiltonian pairs. Under the Miura transformations, we can transform the Hamiltonian pairs into the known
Hamiltonian pairs for the Toda lattice (\ref{todapq}) \cite{A} and modified Volterra lattice (\ref{volt}) \cite{kp85} 
respectively. 
{
  We also show that a specified invariant map of polygons in the centro-affine plane naturally
leads to the integrable map (\ref{todamap}) for the time discretisation of the Toda Lattice.}
{ We will analyze in detail the case of the $2$-homogeneous sphere. We will use normalization equations to obtain Maurer--Cartan invariants and we will prove that they are the classical discrete arc-lengths (the length of the arcs joining vertices) and the discrete curvatures ($\pi$ minus the angle between two consecutive sides of the polygon). We will then write the general formula for invariant evolutions of polygons on the sphere and their invariantizations. We finally find an evolution of polygons whose invariantization is a completely integrable evolution of Volterra type.}

The arrangement of the paper is as follows:
In Section 2 we introduce moving frames and the definitions and calculations we will discretize. In Section 3 we introduce discrete moving frames and
discrete Maurer--Cartan invariants and prove the main Theorems of the first part of the paper concerning these invariants. In Section 4 we begin our study of Lie group invariant discrete evolutions of $N$-gons
in a homogeneous space.  We show the discrete moving frame yields an effective and straightforward reduction or invariantization of the evolution,
to produce discrete analogues of curvature flows. In Section 5 we describe discrete invariant evolutions and demonstrate their integrability. In Sections 5.1 and 5.2 we study the centro-affine and the projective cases, including the completely integrable evolutions of polygons and 
their associated biHamiltonian pair. 
We also briefly describe invariant maps and their invariantizations and  demonstrate that the invariant map in the centro-affine plane
leads to the integrable discretization of the Toda lattice (\ref{todapq}). However,
we leave their thorough study, including a geometric interpretation of integrable maps and more examples of maps with a biPoisson invariantization, for a later paper. Finally, Sections 5.3 describes the more involved example, that of discrete evolutions on the homogeneous sphere.
We conclude with indications of future work.

\section{Background and definitions} 
We assume a smooth Lie group action on a manifold $\M$ given by $G\times \M\rightarrow \M$. { In the significant examples we discuss
in the later sections,  $\M$ will be the set of $N$-gons in a homogeneous space, that is, $N$-gons in $G/H$ where $H$ is a  closed Lie subgroup, or more generally 
$\M = \left(G/H\right)^N$ with the standard action.} { But in what follows this restriction does not need to be in place}.

\subsection{Moving Frames}\label{moving}

We begin with an action of a Lie Group $G$ on a manifold $\M$.

\begin{definition}\label{leftright}
 A \emph{group action} of $G$ on $\M$  is a map $G\times \M\rightarrow \M$, written as $(g,z)\mapsto g\cdot z$, which satisfies either $g\cdot (h\cdot z)=(gh)\cdot z$, called a \emph{left action}, or $g\cdot (h\cdot z)=(hg)\cdot z$, called a \emph{right action}. 
 
 \end{definition} 
 
 \begin{definition}[Invariants] Given a smooth Lie group action $G\times \M\rightarrow \M$, a function $I:\M\to \R$ is an {\em invariant} of the action if $I(g\cdot z) = I(z)$ for any $g\in G$ and any $z\in \M$.
\end{definition}
The set of invariants form an algebra; here we consider local invariants and they will typically be locally smooth. When we talk about generators of the algebra we are referring to functional generators.

Let us write $g\cdot z$ as $\widetilde{z}$ to ease the exposition in places.
Further, we assume the action is {\em free\/} and {\em regular\/} in some domain $\Omega\subset \M$, which means, in effect, that for every $x\in\Omega$ there is a neighbourhood $\mathcal{U}\subset\Omega$ of $x$ such that:
 
\begin{enumerate}
\item
the intersection of the orbits with  $\mathcal{U}$ have the dimension of the group $G$ and further foliate $\mathcal{U}$;
\item
there exists a submanifold $\mathcal{K}\subset \mathcal{U}$ that intersects the orbits of $\mathcal{U}$ transversally, and the intersection of an orbit of $\mathcal{U}$ with $\mathcal{K}$ is a single point. This submanifold $\mathcal{K}$ is known as the \emph{cross-section} and has dimension equal to $\mbox{dim}(\M)-\mbox{dim}(G)$;
\item
if we let $\mathcal{O}(z)$ denote the orbit through $z$, then the element $h \in G$ that takes $z\in \mathcal{U}$ to $k$, where $\{k\}=\mathcal{O}(z)\cap \mathcal{K}$, is unique.
\end{enumerate}

\begin{figure}[tbh]
\centering
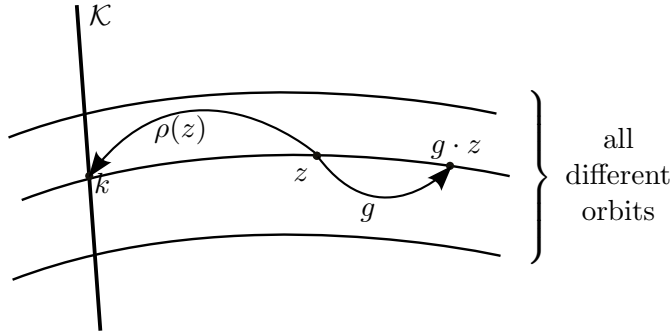
\caption{\label{MMWfig1} The definition of a right moving frame for a free and regular group action. It can be seen that
 $\rho(g\cdot z)= 
\rho(z)g^{-1}$ (for a left action). A left moving frame is obtained by taking the inverse of $\rho(z)$.}
\end{figure}

Under these conditions,  we can make the following definitions.
\begin{definition}[Moving frame]
Given a smooth Lie group action $G\times \M\rightarrow \M$, 
a {\em moving frame} is an equivariant map $\rho:\mathcal{U}\subset \M \rightarrow G$. We say $\mathcal{U}$ is the domain of the frame. 
\end{definition}
Given a cross-section $\mathcal{K}$ to the orbits of a free and regular action, we can define the map $\rho:\mathcal{U} \rightarrow G$ such that $\rho(z)$ is the unique element in $G$ which satisfies
$$\rho(z)\cdot z=k,\qquad \{k\}=\mathcal{O}(z)\cap \mathcal{K},$$
see Figure \ref{MMWfig1}.
We say $\rho$ is the \emph{right moving frame} relative to the cross-section $\mathcal{K}$, { and $\mathcal{K}$ provides the \textit{normalization} of $\rho$}. {This process is familiar to many readers: it is well known that if we translate a planar curve so that a point $p$ in the curve is moved to the origin, and we rotate it so that the curve is tangent to the $x$-axis, the second term in the Taylor expansion at $p$ is the Euclidean curvature at $p$. The element of the Euclidean group taking the curve to its normalization is indeed a right moving frame.}

By construction, we have for a left action and a right moving frame that $\rho(g\cdot z)=\rho(z)g^{-1}$  
so that $\rho$ is indeed equivariant.  A \textit{left} moving frame is the inverse of a right moving frame, so for a left action and a left moving frame, the equivariance is $\rho(g\cdot z) = g \rho(z)$.
The cross-section $\mathcal{K}$ is not unique, and is usually selected to simplify the calculations for a given application. 
Typically, moving frames exist only locally, that is, in some open domain in $M$. 
In what follows, we conflate $\mathcal{U}$ with $M$ to ease the exposition. In applications however, the choice of domain may be critical.

In practice, the procedure to find a right moving frame is as follows:

\begin{enumerate}
\item
define the cross-section $\mathcal{K}$ to be the locus of the set of equations $\psi_i(z)=0$, for $i=1,...,r$, where $r$ is the dimension of the group $G$;
\item
find the group element in $G$ which maps $z$ to $k \in \mathcal{K}$ by solving the \emph{normalization equations},
$$\psi_i(\widetilde{z})=\psi_i(g\cdot z)=0,\qquad i=1,...,r.$$
\end{enumerate}
Hence, the frame $\rho$ satisfies $\psi_i(\rho(z)\cdot z)=0$, $i=1,...,r$.\smallskip

\begin{remark} In practice, both left and right actions occur naturally, as do left and right moving frames, and the calculations can be considerably easier in one handedness than the other, so both occur in the examples. However, since the handedness can be changed by
taking inverses, there is no need to consider both in a theoretical development. We will usually work with  right and left moving frames for a left action.      \end{remark}
In what follows we will use the terms right and left frames to designate left and right moving frames since this creates no confusion.
Invariants of the group action are easily obtained. 
\begin{definition}[Normalized invariants]
Given a left or right action $G\times \M\rightarrow \M$ and a
 right  frame $\rho$, the normalized invariants are the coordinates of $I(z)= \rho(z)\cdot z$.
\end{definition}
Indeed, for a left action we have $$  I(g\cdot z)=\rho(g\cdot z)\cdot g\cdot z= \rho(z) g^{-1} g \cdot z = \rho(z)\cdot z = I(z).$$

The normalized invariants are important because any invariant can be written in terms of them; this follows from the following:
\begin{theorem}[Replacement Rule]\label{Reprule}
If $F(z)$ is an invariant of the action $G\times \M\rightarrow \M$, and $I(z)$ is the normalized invariant for a moving frame $\rho$ on $\M$, then
$F(z)=F(I(z))$.
\end{theorem}
The theorem is proved by noting that (for a right moving frame)  $F(z)=F(g\cdot z)=F(\rho(z)\cdot z)=F(I(z))$
where the first equality holds for all $g\in G$ as $F$ is invariant, the second by virtue of setting $g=\rho(z)$, and the third
by the definition of $I(z)$.

The Replacement Rule shows that the normalized invariants form a set of generators for the algebra of invariants. Further,
if we know a sufficient number of invariants, for example,  they may be known historically or through physical considerations,
then the Replacement Rule allows us to calculate the normalized invariants without knowing the frame.

We can apply this theory to \textit{product (also called diagonal) actions}, the action with which our paper is concerned. This move has its shortcomings, as we see next.

Given a Lie group action
$ G\times M \rightarrow M$, $(g,z)\mapsto g \cdot z$, the product action is
$$ G\times (M\times M\times \cdots \times M),\qquad (g,(z_1, z_2, \dots , z_N))\mapsto (g\cdot z_1, g\cdot z_2, \dots, g\cdot z_N).$$
In this case $\M = M^N = M\times\dots \times M$ and the normalized invariants are {  the invariantized components}\footnote{In the literature the invariantization of a component $z_i$ is sometimes denoted as $I(z_i)$,  $\iota(z_i)$ or $\bar{\iota}(z_i)$, here we use $I_i(z)$} { of  $I(z)$; we set} $I(z) ={ (I_1(z), \dots, I_N(z))}$, and the Replacement Rule
for the invariant
$F(z_1, z_2, \dots, z_N)$  has the form:
$$
F(z_1, z_2, \dots, z_N)
{ =F(I_1(z), I_2(z), \dots, I_N(z))}
$$
{ These were called {\it joint invariants} in \cite{OJS}.} 
\begin{remark} From now on the manifold where $G$ acts will be the product $\M = M\times M\dots \times M$, hence, questions like freedom of the action, etc, will refer to the diagonal action on the product. We will also assume that the number of copies $N$ of the manifold $M$ is high  enough to guarantee that the action is free. In fact, there is a minimal number that achieves this (called the {\it stabilization order}), the proof of this result can be found in \cite{boutin}.
\end{remark}
The following simple scaling and translation group action on $\mathbb{R}$ will be developed as our main expository example, as the calculations are
easily seen.
\vskip 1ex
\begin{example}\label{mainexposex}\end{example} Let $G=\mathbb{R}^{+}\ltimes\mathbb{R}=\{ \left(\lambda, a\right)|, \lambda >0, a\in \mathbb{R}\}$
act on $M=\mathbb{R}$ as
$$ z\mapsto \lambda z + a.$$ The product action is then given by $z_n \mapsto \lambda z_n + a$ for  all $n = 1,\dots, N$.
We write this action as a left linear action as
$$\left(\begin{array}{cccc} z_1 & z_2 & \cdots & z_N\\1&1&\cdots &1\end{array}\right)\mapsto
\left(\begin{array}{cc}\lambda & a\\0&1\end{array}\right) \left(\begin{array}{cccc} z_1 & z_2 & \cdots & z_N\\1&1&\cdots &1\end{array}\right).$$
There are two group parameters and so we need two independent normalization equations. We may set
$$ g\cdot z_1=0,\qquad g\cdot z_2 =1$$ which are solvable if $z_1\ne z_2$, and this then defines the domain of this frame. In matrix form the right frame is,
$$\rho(z_1, z_2, \cdots, z_N)=\left(\begin{array}{cc}-\displaystyle\frac1{z_1-z_2}& \displaystyle\frac{z_1}{z_1-z_2}\\0&1\end{array}\right).$$
Equivariance is easily shown. For $h=h(\mu, b)$:
$$\begin{array}{rcl}\rho(h \cdot z_1, h\cdot z_2, \dots, h\cdot z_N)&=&
\left(\begin{array}{cc}-\displaystyle\frac1{\mu(z_1-z_2)}& \displaystyle\frac{\mu z_1+b}{\mu(z_1-z_2)}\\0&1\end{array}\right)\\[20pt]
&=&\left(\begin{array}{cc}-\displaystyle\frac1{z_1-z_2}& \displaystyle\frac{z_1}{z_1-z_2}\\0&1\end{array}\right)
\left(\begin{array}{cc} \displaystyle\frac1{\mu} & -\displaystyle\frac{b}{\mu}\\0&1\end{array}\right)\\[20pt]
&=& \rho(z_1,z_2,\dots, z_N) h(\mu,b)^{-1}
\end{array}$$
The normalized invariants are then the components of
$$\begin{array}{l}
\left(\begin{array}{cc}-\displaystyle\frac1{z_1-z_2}& \displaystyle\frac{z_1}{z_1-z_2}\\0&1\end{array}\right)
\left(\begin{array}{ccccc} z_1 & z_2 & z_3&\cdots & z_N\\1&1&1&\cdots &1\end{array}\right)\\[20pt]
\ = \left(\begin{array}{ccccc} { I_1(z) }&{ I_2(z)} &{ I_3(z)}& \cdots &{ I_N(z)}\\1&1&1&\cdots &1\end{array}\right)\\[20pt]
\ = \left(\begin{array}{ccccc} 0 & 1 &\displaystyle\frac{z_3-z_1}{z_2-z_1}& \cdots  &\displaystyle\frac{z_N-z_1}{z_2-z_1}\\1&1&1&\cdots &1\end{array}\right),
\end{array}$$
noting that ${ I_1(z)=0}$ and ${ I_2(z)=1}$ are the normalization equations.
Any other invariant can be written in terms of these by the Replacement Rule. For example, it is easily verified that the invariant
$$\left(\frac{z_5-z_4}{z_N-z_6}\right)^2 = \left({ \frac{I_5(z)-I_4(z)}{I_N(z)-I_6(z)}}\right)^2
=\left(\frac{\displaystyle\frac{z_5-z_1}{z_2-z_1}-\displaystyle\frac{z_4-z_1}{z_2-z_1}}{\displaystyle\frac{z_N-z_1}{z_2-z_1}-\displaystyle\frac{z_6-z_1}{z_2-z_1}}\right)^2.$$

The above example shows both the power and the limits of the moving frame. 
In applications involving discrete systems, use of the shift operator $\T$ taking $z_n$ to $z_{n+1}$ is central to the calculations and formulae involved. However, it can be seen in the above example that
${  \T(I_n(z))\ne I_{n+1}(z)}$ and that instead the Replacement Rule will lead to expressions for ${ \T(I_n(z))}$ that are as complicated in their new formulae as in the original co-ordinates. By contrast, our discrete moving frame theory yields a single expression which, together with its shifts, generates all invariants.

{ As motivation} for our discrete moving frame, the next example will show how an \textit{indexed sequence} of moving frames arises naturally in a discrete variational problem. This sequence of frames arises when Noether's conservation laws for a discrete variational problem with a Lie group symmetry, is written, as far as possible, in terms of invariants. As for the smooth case, the non-invariant part turns out to be equivariant with respect to the group action, and so  by definition forms a moving frame, indexed by $n$ in this discrete case. These frames are not constructed by normalisation equations and the calculation of the conservation laws makes no use of the concepts of moving frame theory. Nevertheless, it is fascinating that frames appear, calculated ``for free" by the 
formulae for the laws.
\begin{example}\label{disvarScalTransDisFrame}\end{example} Consider the discrete variational problem
$$\mathcal{L}[\mathbf{z}]=\sum L_n(\mathbf{z})=\sum\frac12 J_{n}^2=\sum \frac12 \left(\displaystyle\frac{z_{n+2}-z_{n+1}}{z_{n+1}-z_n}\right)^2$$
for which it is desired to find { the sequence $(z_n)$ which minimises $\mathcal{L}[\mathbf{z}]$, possibly subject to certain boundary conditions.} 
We note that $J_{n+k}=(z_{n+k+2}-z_{n+k+1})/(z_{n+k+1}-z_{n+k})$ is invariant under the scaling and translation action, $z_n\mapsto \lambda z_n + a$ (see Example \ref{mainexposex})
and that $\T(J_{n+k})=J_{n+k+1}$  where $\T$ is the shift operator.
The discrete Euler Lagrange equation for a second order  Lagrangian in a single discrete variable is  (cf.\ \cite{HyMan})
$$0=\frac{\partial L_n}{\partial z_n} + \T^{-1}\frac{\partial L_{n}}{\partial z_{n+1}}+ \T^{-2}\frac{\partial L_{n}}{\partial z_{n+2}}$$
which { for this example can, after some calculation,}  be written as
$$0=J_{n+1}^2-J_{n}^2-J_{n}^3+J_{n}J_{n-1}^2$$
 after shifting to balance the indices.
The Lie group invariance of the summand $L_n$ implies there are two conservation laws (in this case, first integrals) arising from 
the discrete analogue of Noether's Theorem \cite{HerHick,HyMan}. { The general formulae for these are complicated and we do not record them here.} The
two first integrals can be written in matrix form as
$$\left(\begin{array}{c} c_1 \\ c_2\end{array}\right) = A_n(\mathbf{z}) \mathbf{v}_n(J)=\left(\begin{array}{cc} 
-\displaystyle\frac{1}{z_{n+1}-z_{n}}&0 \\[20pt]-\displaystyle\frac{z_{n+1}}{z_{n+1}-z_{n}} &1 \end{array}\right)
\left(\begin{array}{c} \displaystyle\frac{J_{n+1}^2-J_{n}^2}{J_{n}}  \\[20pt] J_{n}^2 \end{array}\right),$$
where this defines the matrix $A_n(\mathbf{z})$ and the vector of invariants $ \mathbf{v}_n(J)$, and the $c_i$ are the constants of integration.
It can be seen that in this example, the Euler Lagrange equation gives a recurrence relation for $J_n$ and that once this is solved,
the conservation laws yield (if $c_1\ne 0$),
$$ z_{n+1} = \frac{c_2 - J_{n}^2}{c_1 }.$$

The matrix $A_n(\mathbf{z})$ is equivariant under the group action, for each $n$, indeed
$$  A_n(g\cdot\mathbf{z})=
\left(\begin{array}{cc} 
-\displaystyle\frac{1}{\mu(z_{n+1}-z_{n})}&0 \\[15pt]-\displaystyle\frac{z_{n+1}+b/\mu}{z_{n+1}-z_{n}} &1 \end{array}\right)
=\left(\begin{array}{cc} \displaystyle\frac{1}{\mu}&0 \\[15pt]\displaystyle\frac{b}{\mu} & 1\end{array}\right)
\left(\begin{array}{cc} 
-\displaystyle\frac{1}{z_{n+1}-z_{n}}&0 \\[15pt]-\displaystyle\frac{z_{n+1}}{z_{n+1}-z_{n}} &1 \end{array}\right)
$$
and hence each $A_n(\mathbf{z})$ is a left moving frame for the group $G=\mathbb{R}^{+}\ltimes\mathbb{R}$, 
albeit for a different representation
for this group. In fact, this is the Adjoint representation of $G$, placing this result in line with Theorems on the equivariance
of Noether's conservation laws for smooth systems \cite{GonMan, mansfield}.

\begin{figure}[tbh]
\centering
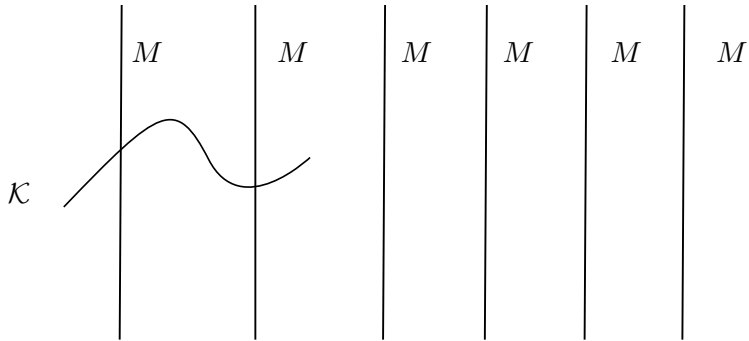
\caption{\label{MMWfig2} The location of the arguments occurring in the normalization equations  for a moving frame for the product action on $M\times M\times \cdots\times M$, shown here as
disjoint copies of $M$ for clarity, for the Example \ref{mainexposex}.}
\end{figure}

\begin{figure}[tbh]
\centering
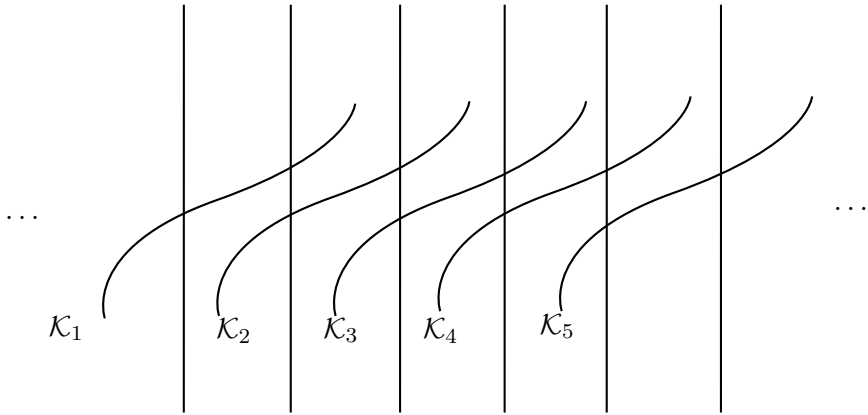
\caption{\label{MMWfig3} The sequence of moving frames for the product action on $M\times M\times \cdots\times M$ in Example 
\ref{disvarScalTransDisFrame}, shown here as
disjoint copies of $M$ so the location of the arguments occurring in the normalization equations can be easily seen. 
For this example, each $\mathcal{K}_i$ is a shift of the previous. This sequence of moving frames is an example of a \textit{discrete moving frame}.}
\end{figure}

In Figure \ref{MMWfig2} is shown, schematically, the location of the arguments occurring in the normalization equations for the frame of 
Example \ref{mainexposex}. By contrast, in Figure \ref{MMWfig3} is shown, schematically, the location of the arguments occurring in the normalization equations for the sequence of frames calculated in Example \ref{disvarScalTransDisFrame}. In fact, for that example, each of the $\mathcal{K}_i$ will be shifts of each other.

In the next section, we will define a discrete moving frame to be a sequence of moving frames with a nontrivial intersection of domains, and will 
explore the properties of the structure of the algebra of invariants that arises. 

\section{Discrete moving frames and the example of twisted $N$-gons in a homogeneous manifold $G/H$}
\subsection{Discrete moving frames and discrete invariants}
In this section we will state the definition of discrete group-based moving frame along $N$-gons and show that a parallel normalization process will produce right discrete moving frames. A discrete moving frame gives many sets of generators for the algebra of invariants under the action of the group. Under conditions that are satisfied in our examples, we will show
recursion relations between these and find a small, useful set of generators which we will denote as \textit{discrete Maurer--Cartan invariants}. 


Our next definition represents the discrete analog of the group-based moving frame we described in Section 2. 

\begin{definition}[Discrete moving frame] Let $G^{N}$ denote the Cartesian product of $N$ copies of the group $G$. Allow $G$ to act on the left on $G^{N}$ using the diagonal action
\(
g\cdot (g_r) = (gg_r).
\)
We also consider what we call the ``right inverse action"  $g\cdot (g_r) = (g_r g^{-1})$. (We note these are both {left} actions according to Definition (\ref{leftright}).)
We say a map 
\[
\rho:M^N \to G^N
\] 
is a left (resp. right)  {\it discrete moving frame} if $\rho$ is equivariant with respect to the { diagonal} action of $G$ on $M^N$ and the left (resp. right inverse) { diagonal} action of $G$ on $G^N$. Since $\rho((x_r))\in G^N$, we will denote by $\rho_s$ its $s$th component, that is $\rho = (\rho_s)$, where $\rho_s((x_r)) \in G$ for all $s$. Equivariance  means,
\[
\rho_s(g\cdot (x_r)) = \rho_s((g\cdot x_r)) = g\rho_s((x_r)) \hskip 2ex(\mathrm{resp.} \hskip 1ex\rho_s((x_r))g^{-1})
\]
for every $s$. Clearly, if $\rho = (\rho_s)$ is a left moving frame, then $\widehat \rho = (\rho_s^{-1})$ is a right moving frame.
\end{definition}

 \begin{definition} Let $F:M^N \to \R$ be a function defined on $N$-gons. We say that $F$ is a {\it discrete invariant} if
 \begin{equation}\label{invdef}
 F((g\cdot x_r)) = F((x_r))
 \end{equation}
 for any $g\in G$ and any $(x_r)\in M^N$.
 \end{definition}
 Notice that the quantities $\rho_s((x_r)) \cdot x_k = I_k^s$ are always invariant as we can readily see from
 $g\cdot  I_k^s=\rho_s(g\cdot(x_r)) \cdot \left(g\cdot x_k\right)=\rho_s((x_r))  g^{-1}g\cdot x_k=\rho_s((x_r)) \cdot x_k= I_k^s$.
 \begin{proposition}[$s$-Replacement rule]\label{repruledis} If $\rho$ is a right moving frame, and $F((x_r))$ is any invariant, then
 \begin{equation}\label{subst}
F((x_r)) = F((\rho_s\cdot x_r))
\end{equation}
 Further,
  the invariants $\rho_s\cdot x_r = I_r^s$ {\bf with $s$ fixed}, $r = 1,\dots N$ generate all other discrete invariants. We call the  $I_r^s$ with $s$ fixed,  
  {\em the $s$-{basic} invariants}.
 \end{proposition}
 \begin{proof} This is the same as Theorem \ref{Reprule} for each $s$. From (\ref{invdef}) we need only choose $g = \rho_s$ in (\ref{invdef}). 
 In particular,  Equation (\ref{subst}) means $F$ can be written as a function of the $s$-basic invariants: one 
 merely needs to substitute $x_r$ by $\rho_s\cdot x_r$ for $s$ fixed. { For those familiar with moving frames, this proposition follows directly from the fact that each $\rho_s$ is a moving frame and $\rho$ is a collection of $N$ moving frames.}
\end{proof}
Locally, a discrete moving frame is uniquely determined by the choice of cross-sections $\mathcal{K}_s$ to the group orbit through $(x_r)$, as in figure \ref{MMWfig3}. Often these sections are shifts of a first section, but that need not be the case in general. The right moving frame component $\rho_s$  is the unique element of the group that takes $(x_r)$ to the cross section $\mathcal{K}_s$. Notice that $\rho_s$ will depend, in general, on several different coordinates, not just $x_s$.

\begin{theorem} Let $\mathcal{K}$ be $N$ local choices $\mathcal{K}_1, \dots \mathcal{K}_N$ of cross-sections to the orbit of $G$ through $(x_r)$. Let $\rho = (\rho_s)\in G$ be uniquely determined by the condition
\begin{equation}\label{sectioneq}
\rho_s\cdot (x_r) \in \mathcal{K}_s, 
\end{equation}
for any $s$. Then $\rho = (\rho_s((x_r)))$ is a right moving frame along the $N$-gon $(x_r)$. 
\end{theorem}
\begin{proof} The proof is entirely analogous to that of a single frame. We will denote by $\rho_s((x_r))$ the unique element of the group determined by equations (\ref{sectioneq}).
On the other hand $\rho_s((x_r)) g^{-1}$ also satisfies those equations, and so $\rho_s(g\cdot(x_r)) = \rho_s((x_r)) g^{-1}$. Therefore $\rho$ is a right discrete moving frame. 
\end{proof}

As in the continuous case, a practical way to obtain a right discrete moving frame is using  a coordinate cross-section, defined as follows. 
Denote by $x^i$ the coordinates of $x\in M$. Let $\rho = (\rho_s((x_r)))\in G$ be uniquely determined by a number of equations of the form
\begin{equation}\label{normeq}
\left(\rho_s((x_r))\cdot x_k\right)^i = (c^s_k)^i, 
\end{equation}
for any $s$,  $k$ in a set that depends on $s$, and some chosen values of $i$ that depend on $s$ and $k$, and where  $(c_r^s)^i\in \R$.  
{ From now on those $(x_r)$ for which the map $g \to (g\cdot x_r)^i$ has full rank for the choice of $i$'s defining the moving frame will be called {\it regular points}}. We usually refer to  the $(c_k^s)^i\in M$ as the {\it normalization constants}.

{ \begin{remark} In some of our examples, and indeed for all examples we envisage, the normalisation equations for the $\rho_r$ are shifts of each other. However, it is interesting that this is not necessary to produce the recursion relation (\ref{recursion}) we derive in Section \ref{disMCinvs}. When normalization equations are shifts of each other the construction of a discrete moving frame depends only on a single function of several variables and the different components of the moving frame are generated by evaluating this function on the corresponding vertices. \end{remark}}


 \subsection{Discrete Maurer--Cartan invariants}\label{disMCinvs}
A discrete moving frame {on $M^N$} provides $N$ different sets of generators, since each $\rho_s$ generates a complete set.
With this abundance of choices one is hard pressed to chose certain distinguished invariants, and a small group that will generate all others. Although a correct choice will depend on what you would like to use the invariants for, a good choice in our case are what we will denote the discrete Maurer--Cartan invariants. They are produced by the discrete equivalent of the Serret--Frenet equations. From now on we will assume that $G \subset \GL(m, \R)$ and so $K_r$ are represented by matrices; this is not necessary, but it is convenient.

The following definition appeared in \cite{OST}.

\begin{definition}[Twisted $N$-gon]
{\it A twisted $N$-gon} in a manifold $M$ is a map $\phi:\Z\to M$ such that for some fixed $g\in G$ we have $ \phi(k+N) = g\cdot \phi(k)$ for all $k\in \Z$. (Recall $\cdot$ represents the action of $G$ on $M$ { given by left multiplication on representatives of the class}.) The element $g\in G$ is called {\it the monodromy} of the gon. 
\end{definition}
\begin{remark} The choice of twisted $N$-gons is one among several ways to make the shift operator well-defined. Because of the need to apply the shift operator freely on the $N$-gons, one is forced to either consider infinite polygons, or to impose some kind of periodicity condition. Although one could merely choose closed gons, there is an important reason to work with twisted $N$-gons instead: twisted $N$-gons will give rise to $N$-periodic invariants; however, a general set of periodic invariants will not, in general, be associated to $N$-periodic gons, but to twisted ones. Nevertheless, for much of what follows we can also consider infinite gons and we will occasionally consider the infinite case in the applications.
\end{remark} 
The space of twisted $N$-gons in $M$ can be identified with the Cartesian product of $N$ copies of the manifold $M$, and hence our previous theory applies.

\begin{definition} 
Let ${ (\rho_s)}$ be a left (resp. right) discrete moving frame evaluated along a twisted $N$-gon, the element of the group
\[
K_s = \rho^{-1}_s\rho_{s+1} \hskip 2ex (\hbox{resp.}\hskip 1ex \rho_{s+1}\rho_s^{-1})
\]
is called the left (resp. right) $s$-Maurer--Cartan matrix for $\rho$. We will call the equation $\rho_{s+1} = \rho_s K_s$ the discrete left $s$-Serret--Frenet equation.
\end{definition}
One can directly check that if $K_s$ is a left Maurer--Cartan matrix for the left frame $(\rho_s)$, then $K_s^{-1}$ is a right one for the
right frame $\widehat{\rho}=(\rho^{-1}_s)$, and vice versa.
 
The equivariance of $\rho$ immediately yields that the $K_s$ are invariant under the action of $G$. We show next how the components of the Maurer--Cartan
matrices generate the algebra of invariants by exhibiting recursion relations between them and the normalized invariants, for our expository example.

\vskip 1ex
\noindent\textbf{Example \ref{mainexposex} cont.}\  Recall the group is $G=\mathbb{R}^+\ltimes \mathbb{R}$ acting  on $\mathbb{R}$ as
$\widetilde{z}=(\lambda, a)\cdot z = \lambda z +a$, which is represented so that the action is left (multiplication) as
$$ \left(\begin{array}{cc} \lambda & a\\ 0&1\end{array}\right)\left(\begin{array}{c} z \\ 1\end{array}\right) =
\left(\begin{array}{c} \lambda z +a \\ 1\end{array}\right).$$
If we take the normalization equations of $\rho_s$ to be $\widetilde{z}_s=1$ and $\widetilde{z}_{s+1}=0$, then the right discrete frame is
$$\rho_s = \left(\begin{array}{cc}-\displaystyle\frac{1}{z_{s+1}-z_s} &\displaystyle\frac{z_{s+1}}{z_{s+1}-z_s} \\ [12pt] 0&1 \end{array}\right).$$
By definition, $\rho_s\cdot z_r = I^s_r=-(z_{r}-z_{s+1})/(z_{s+1}-z_s)$ and then\footnote{We calculate the left Maurer--Cartan matrix for the left frame
$\rho^{-1}$ as the calculations are simpler.}
\begin{equation}\label{firstKeq}K_s = \rho_s\rho_{s+1}^{-1} = \left(\begin{array}{cc} -I^{s}_{s+2} & I^{s}_{s+2}\\0&1\end{array}\right)\end{equation}
as can be verified directly. We do not need, however, to have solved for the frame to obtain this result; it can be obtained directly by
noting that we have both
\begin{equation}\label{TwoRhoEqs} \rho_{s+1}\left(\begin{array}{cc} z_{s+1} & z_{s+2}\\ 1&1\end{array}\right) = \left(\begin{array}{cc} 1 &0\\ 1&1\end{array}\right),
\qquad 
\rho_{s}\left(\begin{array}{cc} z_{s+1} & z_{s+2}\\ 1&1\end{array}\right) = \left(\begin{array}{cc} 0 &I^{s}_{s+2}\\ 1&1\end{array}\right)
\end{equation}
or
$$\rho_s\rho_{s+1}^{-1}= \left(\begin{array}{cc} 0 &I^{s}_{s+2}\\ 1&1\end{array}\right)\left(\begin{array}{cc} 1 &0\\ 1&1\end{array}\right)^{-1},$$
verifying the result in (\ref{firstKeq}).
Note that for the second calculation, we are using the invariants $I^r_s$ symbolically; if we have not solved for the frame, then we may not know
what these are.

Calculations similar to those of Equation (\ref{TwoRhoEqs}) give recursion relations between the $I^s_r$. For example, we have both
$$\rho_s\left(\begin{array}{cc} z_{s+k} & z_{s+k+1}\\ 1&1\end{array}\right)=\left(\begin{array}{cc} I_{s+k}^s&I_{s+k+1}^s\\1&1\end{array}\right),
\quad \rho_{s+k}\left(\begin{array}{cc} z_{s+k} & z_{s+k+1}\\ 1&1\end{array}\right)=\left(\begin{array}{cc} 1&0\\1&1\end{array}\right)
$$
so that
$$K_s K_{s+1}\cdots K_{s+k-1}=\rho_s\rho_{s+k}^{-1}=\left(\begin{array}{cc} I_{s+k}^s- I_{s+k+1}^s&I_{s+k+1}^s\\ 0&1\end{array}\right).$$
However, we also have by direct calculation using our  $K_i$, for example,
$$K_s K_{s+1}=\left(\begin{array}{cc} I_{s+2}^sI_{s+3}^{s+1} & I_{s+2}^s(1-I_{s+3}^{s+1})\\ 0 &1\end{array}\right)$$
and thus for $k=2$ we have $I_{s+3}^{s}=I_{s+2}^s(1-I_{s+3}^{s+1})$.
  In this way, it is possible to obtain all the $I_r^s$ from the components
of the $K_k$'s, that is, the $I_{k+2}^k$. This smaller set of generating invariants we will denote as the discrete Maurer--Cartan invariants.
\smallskip

\begin{example}\label{ex3dot6}\end{example} Our next example is the centro-affine action of $\SL(2, \R)$ on  $\R^2$; that is $\SL(2,\R)$ acts linearly on $\R^2$. 
We may cast the centro-affine plane in the form $G/H$ by identifying $\R^2$ with $SL(2,\R)/H$, where $H$ is the isotropy subgroup of $e_2 = \begin{pmatrix} 0\\ 1\end{pmatrix}$. There is a reason why we are choosing $e_2$ rather than $e_1$; the reader will readily see the connection to Example 
\ref{ex3dot7}, the projective line. In order to find a moving frame we will use the following normalization equations
 \[
 \rho_s\cdot x_s= c_s^s = e_2, \hskip 3ex \rho_s\cdot x_{s+1} = c_{s+1}^s = -\deter{x_s, x_{s+1}} e_1
 \]
 where $\deter{x_s, x_{s+1}} $ is the determinant of the two $2$-vectors; specifically, the second component of $ \rho_s\cdot x_{s+1}$ being zero
 is the normalization equation, while the first is determined by the previous normalizations and the fact that  $\rho_s \in \SL(2, \R)$.  This results in the left moving frame
 \[
 \rho_s^{-1} = \begin{pmatrix}- \frac1{\deter{x_s,x_{s+1}}} x_{s+1} &x_s\end{pmatrix}.
 \]
From here, a complete set of generating invariants will be the components of
\[
\rho_s \cdot x_r = \rho_s x_r = \begin{pmatrix}\deter{x_r, x_s}\\ \frac 1{\deter{x_s, x_{s+1}}} \deter{x_r, x_{s+1}}  \end{pmatrix}
\]
with $s$ fixed and $r = 1, \dots N$; that is, 
$\deter{x_s, x_r}$, $\deter{x_{s+1}, x_r}$ with $s$ fixed and $r = 1, \dots N$.  To see how the $r$-basic invariants will generate the $k$-basic invariants it suffices to recognise the relation
\[
\deter{x_r, x_k} = \deter{\rho_s x_r, \rho_sx_k} = \deter{\begin{array}{cc}\displaystyle\deter{x_r, x_s} & \displaystyle\deter{x_k, x_s}\\ \displaystyle\frac{\deter{x_r, x_{s+1}}}{\deter{x_s, x_{s+1}}}&\displaystyle \frac{\deter{x_k, x_{s+1}}}{\deter{x_s, x_{s+1}}}\end{array}} 
\]
\[
= \frac1{\deter{x_s, x_{s+1}}}\left( \deter{x_r, x_s}\deter{x_k, x_{s+1}} -\deter{x_k, x_s}\deter{x_r, x_{s+1}}\right).
\]
obtained by the Replacement Rule, Proposition \ref{repruledis}.
The left Maurer--Cartan matrix is given by 
\begin{equation}\label{cainv}
\rho_s\rho_{s+1}^{-1} = \begin{pmatrix}k_s^2 & -k_s^1\\ \frac 1{k_s^1} &  0\end{pmatrix}
\end{equation}
where $k_s^2 = \frac{\deter{x_s, x_{s+2}}}{\deter{x_{s+1}, x_{s+2}}}$ and $k_s^1 = \deter{x_s, x_{s+1}}$. 
By calculations similar to those of the previous example, we have that
these invariants, with $s = 1, \dots, N$ generate all other invariants. 
This follows from observing
$$\rho_s\cdot x_r = \rho_s\rho_{s+1}^{-1}\rho_{s+1}\rho_{s+2}^{-1}\cdots \rho_{r-1}\rho_{r}^{-1}\left(\rho_r\cdot x_r\right)
=K_s K_{s+1}\cdots K_{r-1} e_2
$$
and similarly if $r<s$.
In fact, using $k_{s+1}^1$ we could simplify these generators to the simpler set $\deter{x_s, x_{s+1}}, \deter{x_s, x_{s+2}}$, $s= 1,\dots, N$, as expected.
\medskip


Our final expository example before proving our result concerning the discrete Maurer--Cartan invariants is the projective action of $\SL(2, \R)$ on $N$-gons in $\RP^1$.  

\begin{example}\label{ex3dot7}\end{example} Consider local coordinates in $\RP^1$ such that a lift from $\RP^1$ to $\R^2$ is given by $x \to \begin{pmatrix} x\\ 1\end{pmatrix}$. In that case $\RP^1$ can be identified with $\SL(2, \R)/H$, where $H$ is the isotropy subgroup of $x = 0$. The action is given by the fractional transformations
\[
\rho_s \cdot x_r = \begin{pmatrix} a_s&b_s\\ c_s&d_s\end{pmatrix}\cdot x_r = \frac{a_sx_r+b_s}{c_sx_r+d_s}, \qquad a_s d_s-b_s c_s=1.
\]
{ In this particular case one can use a geometric description to find a simple moving frame without resorting to normalization equations (most simple normalization equations would produce a moving frame with no clear geometric meaning)}. Indeed, we can lift $x_s$ to $V_s \in \R^2$ so 
that $\det (V_{s+1}, V_{s}) = 1$ for all $s$. It suffices to define $V_s= t_s\begin{pmatrix} x_s\\ 1\end{pmatrix}$ and solve
\[
t_st_{s+1}\det\begin{pmatrix} x_s& x_{s+1}\\ 1& 1\end{pmatrix} = 1.
\]
{\it If $N$ is not even}, this equation can be uniquely solved for $t_s$, $s=1,\dots,N$, using the twisted condition $t_{N+s} = t_s$ for all $s$. The element 
\[
\rho_s = (V_{s+1}, V_{s})
\]
is clearly a moving frame since the lift of the projective action to $\R^2$ is the linear action. It also satisfies $\rho_s\cdot o=x_s$,
where $o=(0,1)^T$ is the equivalence class of $H$ in $\SL(2, \R)/H$,  where recall $H$ is the isotropy subgroup of $x = 0$. 
Given that $\R^2$ is generated by $V_s, V_{s+1}$ for any choice of $s$, we have that
\[
V_{s+2} = k_s V_{s+1}-V_s
\]
for all $s$ (the coefficient of $V_s$ reflects the fact that $\det(V_{s+1}, V_s) = 1$ for all $s$). From here
\[
\rho_{s+1} = \rho_s \begin{pmatrix} k_s&1\\ -1&0\end{pmatrix}
\]
and so the Maurer--Cartan matrix is given by
\begin{equation}\label{projinv}
K_s = \begin{pmatrix} k_s&1\\ -1&0\end{pmatrix}.
\end{equation}
This invariant appeared in \cite{OST}.
\vskip 1ex

Directly from the definition of Maurer--Cartan matrix one can see that their entries, together with $I_s^s$, are generators for all discrete invariants of $N$-gons.

\begin{proposition}
Let $K = (K_s)$ be the left (resp.\ right) Maurer--Cartan matrix associated to a moving frame $(\rho_s)$.  Then $K$ satisfies the recursion relations
\begin{equation}\label{recursion}
K_s \cdot I_r^{s+1} = I_r^s  \hskip 2ex(\hbox{resp.\ }
K_s\cdot I_r^s = I_r^{s+1})
\end{equation}
for all $s,r$. Furthermore, all discrete invariants are generated by $K$ and $I_s^s$, $s=1, \dots N$.
\end{proposition}
\begin{proof} The proof is a direct straightforward calculation. Since $K_s = \rho_s^{-1}\rho_{s+1}$, $K_s$ will satisfy
\[
K_s\cdot I_r^{s+1} = \rho^{-1}_s\cdot \rho_s\cdot I_r^s = I_r^s
\]
as stated. Since the right Maurer--Cartan matrix for the right frame $\widehat\rho_s$ is the inverse of the left one for
$\widehat\rho_s^{-1}$, it will satisfy $K_s\cdot I_r^s = I_r^{s+1}$.
 
To prove the generating property, let us denote by $I$ any discrete invariant. We know that $I$ will be a function of the basic invariants $I_r^s$ by the Replacement Rule, so it suffices to show that $I_r^s$, $r\ne s$, are all generated by the entries of $K_s$ and $I_s^s$. And this directly follows from the recursion formula. Indeed, we can use the recursion formulas for $K_{s-1}, \dots, K_0$ to find $I_s^{s-1},\dots,I_s^1, I_s^0$ from the values $I_s^s$ recursively. If we invert the recursion equation to read
\[
I_r^{s+1} = K_s^{-1}\cdot I_r^s
\]
then we can use the equations for $K_s, K_{s+1}, \dots K_{N-1}$ to equally generate $I_s^{s+1}$, $I_s^{s+2}, \dots, I_s^N$ from the value $I_s^s$. Therefore, $K_s$ and $I_r^r$ will generate $I_r^s$ for any $s, r=1, \dots, N$.  
\end{proof}

\begin{remark} If the normalization equations guarantee that $\rho$ is uniquely determined, the rank of the map $g \to g\cdot c_r^s$ will have maximum rank for each fixed $s$. Thus, it will equally guarantee that $K = (K_s)$ is uniquely determined by the recursion relations (\ref{recursion}) directly from the normalization constants. Therefore, one is able to find the Maurer--Cartan matrix {\it without knowing explicitly the moving frame}, but merely from the knowledge of the transverse section used to determine it. As we will see in our examples, this simplifies calculations considerably. \end{remark}

In what follows we will assume that our manifold is homogeneous $M = G/H$, with $H$ a closed subgroup. We denote by $o\in G/H$ the distinguished class of $H$.
We will finally show that given $K_s$, and assuming $\rho_s\cdot o = x_s$, for the moving frames generating $K_s$, the $N$-gon is completely determined up to the action of the group.

\begin{proposition} Let $(x_r)$ and $(\hat x_r)$ be two twisted $N$-gons with left moving frames $\rho$ and $\hat \rho$ such that $\rho_s\cdot o = x_s$, $\hat \rho_s \cdot o = \hat x_s$ and $\rho^{-1}_s\rho_{s+1} = \hat\rho_s^{-1}\hat\rho_{s+1} = K_s$. Then, there exists $g\in H$ such that $x_r = g\hat x_r$ for all $r$.
\end{proposition}
 \begin{proof} Notice that from $\rho_{s+1} = \rho_sK_s$ and $\hat\rho_{s+1} = \hat \rho_s K_s$ it suffices to show that there exists $h\in H$ such that $\rho_0 = h\hat\rho_0$. Indeed, if this is true we will have $\rho_s = h\hat\rho_s$ for all $s$, and from here $\rho_s\cdot o = x_s = h\hat\rho_s\cdot o = h\cdot \hat x_s$.
 
 On the other hand, there clearly exists $g\in H$ such that $\rho_0 = g \hat\rho_0$, as it suffices to choose $g = \rho_0\hat\rho_0^{-1}$, and $g\in H$ since both $\rho_0$ and $\hat\rho_0$ leave $o$ invariant. The proposition follows.
 \end{proof}
 
{  Notice that the choice of $o$ in this proposition  is to some extent arbitrary (although it is geometrically the most convenient base point), and the assumption $M = G/H$ is indeed not needed. One could have chosen any fixed point $p \in M$, define $\rho_s\cdot p = x_s$, $\hat\rho_s\cdot p = \hat x_s$, and obtain $g$ in the isotropy subgroup of $p$. One can even choose different base points for $\rho$ and $\hat\rho$ in the same orbit and obtain a general element $g\in G$.}

\section{Lie symmetric evolutions,  maps and their invariantizations}\label{sec4}

In this section we will show how to write {\it any} invariant time evolution of twisted $N$-gons, as well as any invariant map from the space of twisted $N$-gons to itself, in terms of the invariants and the moving frame, in a straight forward and explicit fashion. For simplicity we continue to assume that  our manifold $M$ is
$G/H$ and the normalisation equations include $\rho_s\cdot o=x_s$ (for a left frame) for all $s$. The analogous theorems in the continuous case were published in \cite{M3}.

\begin{definition}

{ An evolution equation is said to have a \textit{Lie group symmetry} if the Lie group action takes solutions to solutions. A recurrence map is said to have a Lie group symmetry if it is equivariant with respect to the action of the group.}
\end{definition}
Such equations and maps
are usually called \textit{invariant}, and we will do so here. However, it is important to note that this does not mean that the equations are comprised of invariants of the action. 
\begin{definition}
We say the evolution 
\begin{equation}\label{invev0}
(x_s)_t = f_s((x_r))
\end{equation}
is an \textit{invariant time evolution} of the twisted $N$-gons  under the action of the group $G$ if its Lie symmetry group is $G$, that is, if $(x_r)$ is a solution, so is $(g\cdot x_r)$ for any $g\in G$. 
\end{definition}

Denote by $\Phi_g:G/H\to G/H$ the map defined by the action of $g\in G$ on $G/H$, that is $\Phi_g(x) = g\cdot x$. Further, denote by
$\mathrm{T}\Phi_g(z)$ the tangent map of $\Phi_g$ at $z\in G/H$. The map $\mathrm{T}\Phi_g(z)$ is written as $\mathrm{T}\Phi_g$ if the base point $z$ is  clear from the context\footnote{Given a manifold $M$ and a map $F:M\rightarrow M$, one standard coordinate free definition of the tangent map  $\mathrm{T}F(z):T_zM \rightarrow T_{F(z)}M$ is as follows. If a path $\gamma(t)\subset M$ satisfies
$\gamma(0)=z$  and the vector $\v\in T_zM={\rm d}/{\rm d}t\big\vert_{t=0} \gamma(t)$, then $\mathrm{T}F(z)(\v)={\rm d}/{\rm d}t\big\vert_{t=0} F\left(\gamma(t)\right)$.}.  In local coordinates $\mathrm{T}\Phi_g$ is given by the Jacobian of $\Phi_g$.  We have by the chain rule both that $(g\cdot x_s)_t = \mathrm{T}\Phi_{g}(x_s)\left((x_s)_t\right)$ and since 
$\Phi_{gh}=\Phi_g\circ\Phi_h$ we have that $\mathrm{T}\Phi_{gh}=\mathrm{T}\Phi_{g}\circ \mathrm{T}\Phi_h$. 
A recent textbook reference on group actions and the associated multi-variable calculus, with and without co-ordinates, is \cite{mansfield}.

Assume we have a discrete \textit{left} moving frame $\rho$ in the neighbourhood of a generic twisted $N$-gon, with a left action, so that
$\rho_s((g\cdot x_r)) = g \rho_s((x_r))$.
We can now state the following Theorem.
\begin{theorem}
Any invariant evolution of the form (\ref{invev0}) can be written as
\begin{equation}\label{invev}
(x_s)_t = \mathrm{T}\Phi_{\rho_s}(o)(\v_s)
\end{equation}
where $o = [H]$, and $\v_s((x_r))\in T_{o} M$ is a vector with invariant components, that is, $\v_s((g\cdot x_r)) = \v_s((x_r))$ for any $g\in G$ and for all $s$. 
\end{theorem}
\begin{proof} 
We  need to prove that
\[
\mathrm{T}\Phi_{\rho_s((x_r))}(o)^{-1} f_s((x_r))
\]
has invariant components under the action of $G$, and call it $\v_s$; recall that the normalisation equations include $\rho_s\cdot o=x_s$ for all $s$, so that $\v_s\in T_{o}M$.
 Indeed, since (\ref{invev0}) is an invariant evolution,
\[
(g\cdot x_s)_t = \mathrm{T}\Phi_{g}(x_s)\cdot(x_s)_t = \mathrm{T}\Phi_{g}(x_s)\cdot f_s((x_r)) = f_s((g\cdot x_r)).
\]
Also, differentiating  $\Phi_{g\rho_s}=\Phi_g\Phi_{\rho_s}$ at $o$ and applying the chain rule yields 
\[
\mathrm{T}\Phi_{\rho_s((g\cdot x_r))}(o) = \mathrm{T}\Phi_{g\rho_s((x_r))}(o) = \mathrm{T}\Phi_{g}(x_s)\circ \mathrm{T}\Phi_{\rho_s((x_r))}(o).
\]
From here
\[
\mathrm{T}\Phi_{\rho_s((g\cdot x_r))}(o)\mathrm{T}\Phi_{\rho_s((x_r))}^{-1}(o) f_s((x_r)) = f_s((g\cdot x_r)), 
\]
which implies that $\mathrm{T}\Phi_{\rho_s((x_r))}(o)^{-1} f_s((x_r))$ is invariant.
\end{proof}
\smallskip

We now consider invariant maps.

\begin{definition}
We say the map 
\begin{equation}\label{invmap0}
F(x_s) = \left(F_s((x_r))\right),
\end{equation}
is an \textit{invariant map} of the twisted $N$-gons  under the action of the group $G$ if its Lie symmetry group is $G$, that is, 
$F((g\cdot x_s)) = g F((x_s))$ for any $g\in G$.
\end{definition}

We have the following theorem. We continue to assume we have a left moving frame and a left action.
\begin{theorem}
If $F$ is an invariant map of the form (\ref{invmap0}), then
\begin{equation}\label{invmap}
F_s((x_r)) =\rho_s((x_r))\cdot z_s((x_r))
\end{equation}
where $z_s ((x_r))\in G/H$ is an invariant  element, that is, $z_s((g\cdot x_r)) = z_s((x_r))$ for any $g\in G$.
\end{theorem}
\begin{proof} It suffices to show that $\rho_s^{-1}\cdot F_s((x_r))$ is invariant and to call it $z_s$. For a left moving frame and a left action we have
\[
\rho_s((g\cdot x_r)) = g \rho_s((x_r))
\]
and by invariance of the map we have
\[
F((g\cdot x_s)) = \left(F_s((g\cdot x_r))\right) = g F((x_s)) = g \left(F_s((x_r))\right),
\]
so that
\[
\rho_s^{-1}((g\cdot x_r)) F_s((g\cdot x_r)) = \rho_s^{-1}((x_r))F_s((x_r)).
\]
\end{proof}

If either a map or a time evolution is invariant, then there is a corresponding induced map or evolution on the invariants themselves. 
The reduction process is at times very involved and time consuming. Here we will describe a simple and straightforward way to find explicitly this so-called {\it invariantization} of the evolution.
{Further, in { Section \ref{sec5}} we detail how, for a proper choice of evolutions, the resulting invariantizations are integrable, in the sense that they can be written in two different ways as a Hamiltonian system, using a Hamiltonian pair. We also illustrate, in Section \ref{invmapaswell}, that some invariantized map also results in a discrete integrable (biPoisson) mapping.
}

\noindent\textit{Assumption.}\/ From now on we will assume to have chosen $\s : M= H \to G$, a section of the quotient $G/H$ such that $\s(o) = e \in G$, where $e$ is the identity. 

\begin{theorem}\label{structeq} Assume we have an invariant evolution of the form (\ref{invev}) and let $\s$ be a section such that $\s(o) = e \in G$. Assume $\rho_r\cdot o = x_r$ and $\rho_r = \s(x_r)\rho_r^H$, where $\rho_r^H \in H$. Then
\begin{equation}\label{invarev}
(K_s)_t = K_sN_{s+1} - N_s K_s
\end{equation}
where $K_s$ is the left Maurer--Cartan matrix and $N_s = \rho_s^{-1}(\rho_s)_t \in \g$. Furthermore, if we split $\g = \m\oplus \h$, where $\g$ is the algebra of $G$, $\h$ is the algebra of $H$ and $\m$ is a linear complement that can be identified with the tangent to the image of the section $\s$, and if $N_s = N_s^\h + N_s^\m$ splits accordingly, then
\begin{equation}\label{Ncond}
N_s^\m = \mathrm{T}\s(o) \v_s.
\end{equation}
\end{theorem}
\begin{proof} The first part of the proof is a straightforward computation
\[
(K_s)_t = \left(\rho_s^{-1}\rho_{s+1}\right)_t = \rho_s^{-1}(\rho_{s+1})_t - \rho_s^{-1}(\rho_s)_t \rho_s^{-1}\rho_{s+1} = K_sN_{s+1}-N_sK_s.
\]
Before proving the second part we notice that if $\s$ is a section, then
\begin{equation}\label{action}
g \s(x) = \s(g\cdot x) h(x,g)
\end{equation}
for some unique $h\in H$. In fact, one can take this relation as defining uniquely the action of the group $G$ on a homogeneous space $G/H$ in terms of the section. We will use this relation shortly.

 Since $\rho_s = \s(x_s)h((x_r))$, by the product and chain rules of differentiation, we have 
\begin{equation}\label{Nsplit}
N_s = \rho_s^{-1}(\rho_s)_t = \left(\mathrm{T}L_{\rho_s^{-1}}\right)\left(\mathrm{T}R_{\rho_s^H}\right) \mathrm{T}\s(x_s)(x_s)_t + \left(\mathrm{T}L^{-1}_{\rho_s^H}\right) T\rho_s^H((x_r))((x_r)_t)
\end{equation}
where $L$ and $R$ signify left and right multiplication and where $\mathrm{T}L$ and $\mathrm{T}R$ are their tangent maps.
Very clearly the second term belongs to $\h$ and so we will focus on the first term.
If we differentiate (\ref{action}) evaluated on $x_r$ we have
\[
\left(\mathrm{T}L_g\right)\mathrm{T}\s(x_r)(x_r)_t = \left(\mathrm{T}R_{h(x_r,g)}\right)\mathrm{T}\s(\Phi_g(x_r))\mathrm{T}\Phi_g(x_r)(x_r)_t + \left(\mathrm{T}L_{\s(g\cdot x)}\right)\mathrm{T}h(x_r, g)(x_r)_t
\]
for any $g$, where $\mathrm{T}h(x,g)$ is the derivative of $h(x,g)$ as a function of $x$. If we now substitute $g = \rho_r^{-1}$ and we use (\ref{invev}), $\s(o) = e$ and $\mathrm{T}\Phi_{\rho_r^{-1}}(x_r) =  \left(\mathrm{T}\Phi_{\rho_r}(o)\right)^{-1}$ we have
\[
\left(\mathrm{T}L_{\rho_r^{-1}}\right) \mathrm{T}\s(x_r)(x_r)_t = \left(\mathrm{T}R_{h(x_r,\rho_r^{-1})}\right) \mathrm{T}\s(o)\v_r+\mathrm{T}h(x_r,\rho_r^{-1})(x_r)_t.
\]
The last term is again an element of $\h$ and we can ignore it. Finally, from (\ref{action}) we see that
\[
(\rho_r^H)^{-1}= \rho_r^{-1}\s(x_r) = \s(\rho_r^{-1}\cdot x_r)h(x_r, \rho_r^{-1}) = \s(o) h(x_r, \rho_r^{-1}) = h(x_r, \rho_r^{-1}).
\]
Therefore
\[
\left(\mathrm{T}R_{h(x_r,\rho_r^{-1})}\right) \mathrm{T}\s(o)\v_r = \left(\mathrm{T}R_{(\rho_r^H)^{-1}}\right)  \mathrm{T}\s(o)\v_r.
\]
Going back to the splitting of $N_r$ in (\ref{Nsplit}) we see that
\[
N_r^\m = \left(\mathrm{T}R_{\rho_r^H}\right) \left(\mathrm{T}R_{(\rho_r^H)^{-1}}\right)  \mathrm{T}\s(o)\v_r = \mathrm{T}\s(o)\v_r
\]
as stated in the theorem.
\end{proof}
It is often the case that conditions (\ref{Ncond}) allow us to solve explicitly for the $N_r$ directly from equation (\ref{invarev}), as we will see in the examples in the next section. Before going there we do a quick description of the invariantization of invariant maps.
We will next prove the analogous result to Theorem \ref{structeq} for invariant maps.
\begin{theorem}
Assume $F$ is an invariant map given as in (\ref{invmap}). Extend this map naturally to functions of $(x_r)$ using the relation $F(\ell((x_r))) = \ell(F((x_r)))$. We will abuse notation and denote both maps with the same letter. Then the map induced on the invariants is given by
\begin{equation}\label{invarmap}
F(K_s) = M_s^{-1} K_s M_{s+1}
\end{equation}
where $M_s = \rho_s^{-1}F(\rho_{s})$. Furthermore, if $\s$ is a section as before, and if $\rho_s = \s(x_s)\rho_s^H$ for some $\rho_s^H \in H$, then $M_s = \s(z_s)M_s^H$ where $M_s^H \in H$.
\end{theorem}
\begin{proof} First of all, notice that since by definition $F(\rho_s((x_r))) = \rho_s(F(x_r))$, $F(\rho_s^{-1})F(\rho_s) = F(\rho_s^{-1}\rho_s) = F(I) = I$ and so $F(\rho_s^{-1}) = F(\rho_s)^{-1}$. From here
\[
F(K_s) = F(\rho_s^{-1})F(\rho_{s+1}) = F(\rho_s)^{-1}\rho_s\rho_s^{-1}\rho_{s+1}\rho_{s+1}^{-1}F(\rho_{s+1}) = M_s^{-1} K_s M_{s+1}.
\]
Also, assuming that $\rho_s = \s(x_s)\rho_s^H$, a direct calculation using (\ref{action}) shows
\begin{eqnarray*}
M_s &=& \rho_s^{-1}F(\s(x_s)\rho_s^H) = \rho_s^{-1}\s(F_s((x_r)))F(\rho_s^H) = \rho_s^{-1} \s(\rho_s\cdot z_s) F(\rho_s^H) \\
&=& \rho_s^{-1}\rho_s\s(z_s)h(z_s, \rho_s)F(\rho_s^H) = \s(z_s)h(z_s, \rho_s)F(\rho_s^H).
\end{eqnarray*}
It suffices to call $M_s^H = h(z_s,\rho_s)F(\rho_s^H)\in H$ to conclude the theorem.
\end{proof}

\section{Completely integrable systems associated to discrete moving frames}\label{sec5}
In this section we will describe invariantizations of general invariant evolutions for our previous two examples (centro-affine and projective) and we will associate completely integrable systems to each one of them by choosing the invariant elements $\v_s$ defining the equation appropriately. In the centro-affine case we  study the invariant maps and produce a well-known biPoisson map resulting from the invariantization of a particular invariant map among polygons obtained when choosing specific values for the invariant elements $z_s$. Finally, we introduce a more complicated example, that of the homogeneous $2$- sphere $S^2 \cong\SO(3)/\SO(2)$.
\subsection{Centro-affine case}\label{centro}
\subsubsection{Invariant evolutions} In the centro-affine case (see Example \ref{ex3dot6}) the action is linear and so $d\Phi_{\rho_s}(o)\v_s = \rho_s\v_s$. The general invariant evolution is given by
\[
(x_s)_t = -\frac{v_s^1}{\deter{x_s, x_{s+1}}} x_{s+1} + v_s^2 x_s =\rho_s^{-1}\begin{pmatrix} v_s^1\\ v_s^2\end{pmatrix}
\]
where $v_s^1, v_s^2$ are arbitrary functions of the invariants previously obtained in (\ref{cainv}) (recall  that $\rho_s$ was in this case a right frame and we need a left one here). If we want to find the evolution induced on $k_s^1$ and $k_s^2$ as in (\ref{cainv}), then we will recall that the space $H$ is the isotropy subgroup of $e_2$, which is the subgroup of strictly lower triangular matrices. A section for the quotient is given by
\begin{equation}\label{affsec}
\s\begin{pmatrix} a\\ b\end{pmatrix} = \begin{pmatrix} b^{-1}& a\\ 0& b\end{pmatrix}
\end{equation}
if $b\ne 0$ (we work in a neighborhood of $o = e_2$). Clearly $\rho_s = \s(x_r) \rho_r^H$ since $\s(x_r)^{-1}\rho_s \in H$. A complement $\m$ to $\h$ is given by the upper triangular matrices. In that case
\[
d\s(o) \v_s = \begin{pmatrix}- v_s^2& v_s^1 \\ 0 & v_s^2\end{pmatrix}
\]
and so 
\[
N_s = \begin{pmatrix}-v_s^2 & v_s^1\\ \alpha_s& v_s^2\end{pmatrix}
\]
where $\alpha_s$ is still to be determined. From equation (\ref{invarev}) we have
\begin{eqnarray*}
&&\begin{pmatrix} (k_s^2)_t&-(k_s^1)_t\\ ((k_s^1)^{-1})_t & 0\end{pmatrix} =
\begin{pmatrix}k_s^2 & - k_s^1\\\frac1{k_s^1}&  0\end{pmatrix} 
\begin{pmatrix}-v_{s+1}^2&v_{s+1}^1\\ \alpha_{s+1}& v_{s+1}^2\end{pmatrix}
 - \begin{pmatrix} -v_{s}^2&v_{s}^1 \\ \alpha_{s}& v_{s}^2\end{pmatrix}
\begin{pmatrix} k_s^2 & -k_s^1\\\frac1{k_s^1}& 0\end{pmatrix}\\
&&\quad\quad =\begin{pmatrix} -k_s^2 v_{s+1}^2- k_s^1 \alpha_{s+1}+v_{s}^2  k_s^2-\frac{v_s^1}{k_s^1} 
& k_s^2 v_{s+1}^1- k_s^1 \v_{s+1} ^2-v_{s}^2  k_s^1\\ 
-\frac{v_{s+1}^2}{k_s^1}-\alpha_{s} k_s^2-\frac{v_s^2}{k_s^1} & \frac{v_{s+1}^1}{k_s^1}+\alpha_{s} k_s^1\end{pmatrix}.
\end{eqnarray*}
The entry $(2,2)$ of this system is given by $0 = k_s^1 \alpha_{s}+(k_s^1)^{-1} v_{s+1}^1$, which allows us to solve for the missing entry
\[
\alpha_s = -\frac{v_{s+1}^1}{(k_{s}^1)^2}.
\]
The other entries give us the evolution of the invariants. These are
\begin{eqnarray}\label{affine}
&&\left( \begin{array}{l}(k_s^1)_t\\(k_s^2)_t\end{array}\right)=\left(\begin{array}{c}
k_s^1 v_{s+1}^2 + k_s^1 v_s^2- k_s^2 v_{s+1}^1\\    
 -k_s^2v_{s+1}^2 - \frac{v_{s}^1}{k_s^1} + k_s^2 v_s^2 + \frac{k_s^1}{(k_{s+1}^1)^2} v_{s+2}^1                                                            
\end{array}\right):=\A \left( \begin{array}{l} v_s^2 \\v_s^1 \end{array}\right),
\end{eqnarray}
where the matrix difference operator is
\begin{eqnarray}\label{opA}
 \A=\left(\begin{array}{cc}
k_s^1 (\T+1)& -k_s^2\T \\    
 -k_s^2 (\T-1) & -\frac{1}{k_s^1} + \frac{k_s^1}{(k_{s+1}^1)^2} \T^{2}                                                            
\end{array}\right)
\end{eqnarray}
and where $\T$ is the shift operator $\T a_s = a_{s+1}$.

Let us define a  diagonal matrix  
\begin{eqnarray}\label{opP}
\mathcal{P}=\left(\begin{array}{cc} (\T-1)^{-1}k_s^1 &0 \\0& -\T^{-1}k_s^1\end{array}\right) . 
\end{eqnarray}
and compute the pseudo-difference operator
\begin{equation}\label{opAP}
\mathcal{AP}=\left(\begin{array}{cc} k^1_s(\T+1)(\T-1)^{-1}k^1_s & k^1_sk^2_s\\
- k^1_sk^2_s & \frac1{k^1_s}\T^{-1} k^1_s - k^1_s\T \frac{1}{k^1_s}
\end{array}\right)
\end{equation}
which is clearly anti-symmetric. We denote it by $\mathcal{H}[k^1_s,k^2_s]$.
\begin{theorem}
The operator $\mathcal{H}[k^1_s,k^2_s] $, given by (\ref{opAP})  is a Hamiltonian
operator. It forms a Hamiltonian pair with Hamiltonian operator 
\begin{eqnarray*}
 \HH_0[k_s^1,k_s^2]=\left(\begin{array}{cc}0 &k_s^1 \\-k_s^1 & 0\end{array}\right) .
\end{eqnarray*}
\end{theorem}
\begin{proof} 
Let us introduce the following Miura transformation
\begin{eqnarray}\label{miura2}
p_s=\frac{k_s^1}{k_{s+1}^1}, \qquad q_s=k_s^2 .
\end{eqnarray}
Its Fr{\'e}chet derivative is
\begin{eqnarray*}
D_{(p_s,q_s)}=\left(\begin{array}{cc} 
\frac{1}{k_{s+1}^1}-\frac{k_s^1}{(k_{s+1}^1)^2} \T&0\\0&1 \end{array} 
\right)=
\left(\begin{array}{cc} \frac{1}{k_{s+1}^1}-p_s \T \frac{1}{k_{s}^1} 
&0\\0&1 \end{array} \right).
\end{eqnarray*}
Under this transformation the operators $\HH_0[k_s^1,k_s^2]$ and $\HH[k_s^1,k_s^2]$ become
\begin{eqnarray*}
&&{\tilde \HH}_0[p_s,q_s]=D_{(p_s,q_s)} \HH_0[k_s^1,k_s^2] D_{(p_s,q_s)}^\star
=\left(\begin{array}{cc} 0 &p_s 
(1-\T)\\-(1-\T^{-1}) p_s&0\end{array} \right) 
\end{eqnarray*}
and 
\begin{eqnarray*}
&&{\tilde \HH}[p_s,q_s]=D_{(p_s,q_s)} \HH[k_s^1,k_s^2] D_{(p_s,q_s)}^\star\\
&& =
\left(\!\!\!\!\begin{array}{cc} \frac{1}{k_{s+1}^1}\!-\!p_s \T 
\frac{1}{k_{s}^1} &0\\0&1 \end{array}\!\!\!\! \right)
\left(\!\!\!\!\begin{array}{cc} k_s^1 
(\T\!+\!1)(\T\!-\!1)^{\!-\!1}k_s^1 & k_s^1 q_s \\ -k_s^1 q_s
& \T^{\!-\!1} p_s\! -\!p_s\T \end{array}\!\!\!\!\right)
\left(\!\!\!\!\begin{array}{cc} \frac{1}{k_{s+1}^1}\!-\! 
\frac{1}{k_{s}^1} \T^{\!-\!1} p_s &0\\0&1 \end{array}\!\!\!\! \right)\\
&& =\left(\begin{array}{cc} p_s \T^{-1} p_s -p_s \T p_s &p_s 
(1-\T)q_s\\-q_s(1-\T^{-1}) p_s& \T^{-1} p_s -p_s\T\end{array} \right) .
\end{eqnarray*}
These two operators form a Hamiltonian pair for the  well-known Toda lattice (\ref{todapq}) in Flaschka coordinates \cite{A,suris03}. Indeed, we have
\begin{eqnarray}\label{toda}
&&\left( \begin{array}{l}(p_s)_t\\(q_s)_t\end{array}\right)=\left(\begin{array}{c}
p_s (q_s-q_{s+1})\\    p_{s-1}-p_s\end{array}\right)={\tilde \HH} \delta q_s
={\tilde \HH}_0 \delta (\left( \frac{1}{2}q_s^2+p_s\right)
\end{eqnarray}
and thus we proved the statement.
\end{proof}
If we take $\left(\! \begin{array}{ll} v_s^2,& v_s^1 \end{array}\!\right)=\left(\! \begin{array}{ll} 0,&
-k_{s-1}^1 \end{array}\!\right)$ in (\ref{affine}), 
then the evolution of the invariants for $k_s^1$ and $k_s^2$ becomes
\begin{eqnarray*}
&&\left( \begin{array}{l}(k_s^1)_t\\(k_s^2)_t\end{array}\right)=\left(\begin{array}{c}
 k_s^1 k_s^2\\    
\frac{k_{s-1}^1}{k_s^1}-\frac{k_{s}^1 }{k_{s+1}^1}                                               
\end{array}\right) =\HH_0 \delta \left( \frac{(k_s^2)^2}{2}+\frac{k_s^1}{k_{s+1}^1}\right)
=\HH \delta  k_s^2,
\end{eqnarray*}
which is an integrable {differential} difference equation and can be transformed into Toda lattice (\ref{toda})
under the transformation (\ref{miura2}).  We summarize our results as a Theorem.

\begin{theorem}{ The evolution of polygons in the centro-affine plane described by the equation
\[
(x_s)_t = \frac{k^1_{s-1}}{|x_s,x_{s+1}|} x_{s+1}
\]
induces a completely integrable system in its curvatures $k_s^1, k_s^2$ equivalent to the Toda Lattice.}
\end{theorem}

The pseudo-difference operator $\P$ is really a formal operator as $\T-1$ is not invertible in the periodic case. It means that we apply the operator only to Hamiltonians whose gradients are in the image of $\T-1$, as is the case here.

\subsubsection{ Invariant maps}\label{invmapaswell}

In the centro-affine case an invariant map is of the form
\[
F(x_s) = -\frac{z_s^1}{\deter{x_s, x_{s+1}}} x_{s+1} + z_s^2x_s = \rho_s^{-1}\begin{pmatrix} z_s^1\\ z_s^2\end{pmatrix}.
\]
Although in this simple case we could directly find the transformation of the invariants, we will follow the process described { in Section \ref{sec4}}. Our example in Section \ref{sphere} will provide a stronger case for the effectiveness of the method. Using section (\ref{affsec}), the transformation under $F$ of $k_s^1$ and $k_s^2$ will be built up using the matrix $M_s = -\rho_{s+1}\rho_s^{-1}$ given by
\[
M_s = \begin{pmatrix} (z_s^2)^{-1} & z_s^1\\ 0&z_s^2\end{pmatrix} \begin{pmatrix} 1&0\\ \alpha_s&1\end{pmatrix}
\]
where $\alpha_s$ needs to be found. The equations relating $K_s$ and $M_s$ are given by
\[
\begin{array}{l}
F\begin{pmatrix} k_s^2&-k_s^1\\ (k_s^1)^{-1} & 0\end{pmatrix}
\\ \quad = \begin{pmatrix} 1&0\\ -\alpha_s&1\end{pmatrix}\begin{pmatrix} z_s^2 & -z_s^1\\ 0&(z_s^2)^{-1}\end{pmatrix}\begin{pmatrix} k_s^2&-k_s^1\\ (k_s^1)^{-1} & 0\end{pmatrix} \begin{pmatrix} (z_{s+1}^2)^{-1} & z_{s+1}^1\\ 0&z_{s+1}^2\end{pmatrix} \begin{pmatrix} 1&0\\ \alpha_{s+1}&1\end{pmatrix}.
\end{array}
\]
The $(2,2)$ entry of this equation will allow us to solve for $\alpha_s$. It is given by
\[
\alpha_{s} \left(\frac1{k_s^1}z_s^1z_{s+1}^1 + z_{s}^2(z_{s+1}^2k_s^1-z_{s+1}^1k_s^2)\right) + \frac{z_{s+1}^1}{z_{s}^2 k_s^1} = 0.
\]
 The other entries of the system will solve for the transformation of $k_s^1$ and $k_s^2$. They are given by
 \begin{eqnarray*}
 F(k_s^1) &=& \frac1{k_s^1}z_s^1z_{s+1}^1+(z_{s+1}^2k_s^1-z_{s+1}^1k_s^2)z_{s}^2\\ 
 F(k_s^2) &=& \left((k_s^2z_{s+1}^1-z_{s+1}^2k_s^1)z_{s}^2 - \frac1{k_s^1}z_s^1z_{s+1}^1\right)\alpha_{s+1} - \frac{z_{s}^1}{z_{s+1}^2k_s^1}+\frac{z_{s}^2k_s^2}{z_{s+1}^2}
 \end{eqnarray*}
 Now we take $z_s^2=\frac{1}{c}$, where $c\neq 0$ is constant and $z_s^1$ 
satisfying the relation
\begin{eqnarray}\label{consz}
\frac1{k_s^1}z_s^1z_{s+1}^1+\frac{k_s^1}{c^2}-\frac{z_{s+1}^1k_s^2}{c}=z_{s+1}^1
\end{eqnarray}
Then the above maps become
  \begin{eqnarray}
  F(k_s^1) &=& z_{s+1}^1\label{mapk1}\\
  F(k_s^2) &=& c\left( \frac{z_{s+1}^1}{k_{s+1}^1}- \frac{\ 
z_{s}^1}{k_s^1}\right)+k_s^2. \label{mapk2}
  \end{eqnarray}
Let $a_s=\frac{k_s^1}{k_{s+1}^1}$, $b_s=k_s^2$ and 
$\beta_s=\frac{k_s^1}{z_{s+1}^1}$. Then using (\ref{mapk1}) and 
(\ref{mapk2}) we have
\begin{eqnarray*}
F(a_s)&=&\frac{F(k_s^1)}{F(k_{s+1}^1)}=\frac{k_s^1}{\beta_s} 
\frac{\beta_{s+1}}{k_{s+1}^1}=a_s \frac{\beta_{s+1}}{\beta_{s}}\\
F(b_s)&=& b_s+c\left( \frac{k_s^1}{k_{s+1}^1 
\beta_{s}}-\frac{k_{s-1}^1}{k_s^1 \beta_{s-1}}\right)
= b_s+c\left( \frac{a_s}{ \beta_{s}}-\frac{a_{s-1}}{\beta_{s-1}}\right)\ .
\end{eqnarray*}
The constraint (\ref{consz}) on $z_s^1$ becomes
\begin{eqnarray*}
\frac{a_{s-1}}{\beta_{s-1}\beta_s}+\frac{1}{c^2}-\frac{b_s}{c\ 
\beta_{s}} =\frac{1}{\beta_{s}},
\end{eqnarray*}
that is,
\begin{eqnarray*}
\beta_s=1+c\ b_s-c^2 \frac{a_{s-1}}{\beta_{s-1}}.
\end{eqnarray*}
Thus we obtain the integrable discretization of the Toda lattice as the formulas 
(3.8.2) and (3.8.3) in \cite{suris03}. { Thus we obtain the following result:
\begin{theorem}{ The invariant map of polygons in the centro-affine plane described by the equation
\[
F(x_s) = -\frac{z_s^1}{\deter{x_s, x_{s+1}}} x_{s+1} + \frac{x_s}{c},
\]
where $c\neq 0$ is constant and $z_s^1$ satisfying (\ref{consz})
induces a completely integrable map in its curvatures $k_s^1, k_s^2$ equivalent to the integrable discretization of the Toda Lattice.}
\end{theorem}}

\subsection{ The projective case}
In the projective example 
the subgroup $H$ is the isotropy subgroup of $0$, that is, $H$ is given by lower triangular matrices. Thus, a section can be chosen to be
\[
\s(x) = \begin{pmatrix} 1&x\\ 0&1\end{pmatrix}.
\]
One can check directly that $\s(x_r)^{-1}\rho_r \in H$ and also
\[
d\s(0) \v = \begin{pmatrix} 0&\v\\ 0&0\end{pmatrix}. 
\]
A general invariant evolution in this case would be of the form
\begin{equation}\label{invev2}
(x_s)_t = d\Phi_{\rho_s}(0) \v_s = \frac1{a_s^2} \v_s
\end{equation}
where $\v_s$ is a function of $(k_r)$ and where $a_s$ is the $(2,2)$ entry of $\rho_s$. The evolution induced on $k_s$ is given by the equation (\ref{invarev}) with 
\[
N_s = \begin{pmatrix} \alpha_s& \v_s\\ \beta_s&-\alpha_s\end{pmatrix}.
\]
We have
\[
\begin{pmatrix}  (k_s)_t&0\\ 0&0\end{pmatrix} = \begin{pmatrix} k_s&1\\ -1&0\end{pmatrix}\begin{pmatrix} \alpha_{s+1}& \v_{s+1}\\ \beta_{s+1}&-\alpha_{s+1}\end{pmatrix} - \begin{pmatrix} \alpha_s& \v_s\\ \beta_s&-\alpha_s\end{pmatrix}\begin{pmatrix} k_s&1\\ -1&0\end{pmatrix}.
\]
The entries $(2,2)$ and $(1,2)$ of this system will allow us to solve for $\beta_r$ and $\alpha_r$. Indeed, the $(2,2)$ entry is given by
\[
0 = -\beta_{s}-\v_{s+1}
\]
and so 
\[
\beta_s = -\v_{s+1}.
\]
The $(1,2)$ entry is given by
\[
\alpha_{s+1}+\alpha_s-k_s\v_{s+1} = 0.
\]
This leads to 
$$
\alpha_s=(\T+1)^{-1} k_s \v_{s+1},
$$
where $\T$ is the shift operator. 
Notice that $\T+1$ is an invertible operator. Assuming $N$ is not even we can solve the equation $(\T+1)\alpha_s = w_s$ to obtain
\[
\alpha_s =\frac{ (-1)^s}{2}\left(  \sum_{i=0}^{s-1}(-1)^{i+1} w_i-\sum_{i=s}^{n-1}(-1)^{i+1}w_i \right).
\]
Finally, the $(1,1)$ entry gives the evolution for $k_s$
\begin{eqnarray}
 (k_s)_t &=& \v_{s}-\v_{s+2}+k_s ( \alpha_{s+1}-\alpha_s)\nonumber\\
&=&\left(\T^{-1}-\T+k_s (\T-1)(\T+1)^{-1} k_s\right) \v_{s+1} \ .\label{invarevproj}
\end{eqnarray}
\begin{theorem}
The anti-symmetric operators $$\HH_1[k_s] =\T-\T^{-1}$$ and $$\HH_2[k_s]=k_s (\T-1)(\T+1)^{-1} k_s$$ form a Hamiltonian pair.
\end{theorem}
\begin{proof}
Let us introduce the following Miura transformation 
\begin{eqnarray}\label{miura}
u_s=\frac{1}{k_s} \ . 
\end{eqnarray}
We compute its Fr{\'e}chet derivative $D_{k_s}=-1/u_s^2$ . 
Under the transformation (\ref{miura}) the operators $\HH_1[k_s]$ and $\HH_2[k_s]$ become
\begin{eqnarray*}
&&\tilde \HH_1[u_s]=D_{k_s}^{-1} \HH_1[k_s] D_{k_s}^{\star -1}  =u_s^2 \left(\T-\T^{-1}\right) u_s^2; \\
&&\tilde \HH_2[u_s]=D_{k_s}^{-1} \HH_2[k_s] D_{k_s}^{\star -1} = u_s (\T-1)(\T+1)^{-1} u_s .
\end{eqnarray*}
These two operators form a Hamiltonian pair for the  well-known modified Volterra lattice (\ref{volt}) \cite{kp85}.
\begin{eqnarray}\label{vol}
 (u_s)_t=u_s^2 (u_{s+1}-u_{s-1})= \tilde \HH_1[u_s] \delta_{u_s} \ln u_s =\tilde \HH_2[u_s]  \delta_{u_s} (u_s u_{s-1}). 
\end{eqnarray}
This implies that $\HH_1[k_s]$ and $\HH_2[k_s]$ form a Hamiltonian pair and the statement is proved.
\end{proof}
Take $\v_{s+1}=\delta_{k_s} \ln k_s$ in (\ref{invarevproj}). The evolution for $k_s$ becomes
\begin{eqnarray*}
(k_s)_t &=& \frac{1}{k_{s-1}}-\frac{1}{k_{s+1}} ,
\end{eqnarray*}
which is an integrable difference equation. Under the transformation (\ref{miura}), it leads to the modified Volterra 
lattice (\ref{vol}).
\vskip 2ex

 At this moment it seems natural to wonder if one could also obtain this modified Volterra evolution in the centro-affine case (\ref{centro}), via a reduction to the space $k_s^1 = 1$ for all $s$. And indeed such is the case.

Assume $k_s^1 = 1$ and choose those evolutions that leave $k_s^1$ invariant. That is, assume $(k_s^1)_t = 0$. Then, from (\ref{affine}) we have
\[
v_{s+1}^2 + v_s^2 - k_s^2 v_{s+1}^1 = 0
\]
and from here
\[
v_s^2 = (\T+1)^{-1} k_s^2 v_{s+1}^1.
\]
Substituting in (\ref{affine}) we get
\[
(k_s^2)_t = -k_s^2\T(\T+1)^{-1} k_s^2 \T v_s^1 - v_s^1 + k_s^2 (\T+1)^{-1} k_s^2 \T v_s^1 + \T^2 v_s^1
\]
\[
= \left[\T-\T^{-1} - k_s^2(\T-1)(\T+1)^{-1}k_s^2\right] \T v_s^1,
\]
which is identical to the projective one with identifications $k_s =  k_s^2$ and $\v_s = -v_s^1$. Hence, choosing $\T v_s^1 = -\delta_{k_s^2} \ln |k_s^2|$ will result in a modified Volterra equation for the centro-affine invariant $k_s^2$, as far as the initial polygon satisfy $k_s^1 = 1$, that is, as far as the area of the parallelogram formed by $x_s$ and $x_{s+1}$ in the plane is equal $1$ for all $s$ and they are properly oriented. Other constant values can be chosen with minimal changes. Notice that this evolution is well-defined in the periodic case.
\vskip 2ex


\subsection{ The homogenous sphere $S^2 \cong\SO(3)/\SO(2)$}\label{sphere}

{In this section we consider invariant evolutions on the homogeneous sphere, $S^2 \cong\SO(3)/\SO(2)$. We first consider a local section using which we can describe our calculations of a discrete frame, the associated discrete Maurer--Cartan matrices, and invariantizations. We then show that, for a certain choice of polygon evolution the resulting curvature flow is integrable, of Volterra type. }

If $G = \SO(3)$ and $H = \SO(2)$, we consider the following splitting of the Lie algebra into subspaces $\so(3) = \m\oplus \h$ with
\begin{equation}
 \begin{pmatrix} 0&y\\ -y^T&0\end{pmatrix}\in \m\hskip 2ex  \begin{pmatrix} A &0\\ 0&0\end{pmatrix} \in \h 
\end{equation}
where $y \in \R^2$ and $A\in\so(2)$. Associated to this splitting we have a local factorization in the group into factors belonging to $H = \SO(2)$ and  $\exp(\m)$. This factorization is given by
\begin{equation}\label{gson+1}
 g = g(\Theta, y) = \begin{pmatrix} \Theta & 0\\ 0&1\end{pmatrix} \begin{pmatrix} I + \cos_y y y^T & \sin_y y\\ -\sin_y y^T & \cos(\norm y)\end{pmatrix}
 \end{equation}
 where $\displaystyle\cos_y = \frac{\cos(\norm{y})-1}{\norm{y}^2}$, $\displaystyle\sin_y = \frac{\sin(\norm{y})}{\norm{y}}$ and $\norm y^2 = y^T y$. The factorization exists locally in a neighborhood of the identity.
 
 Let $\s: S^2 \to \SO(3)$ be the section defined by the exponential, that is
 \[
 \s(x) =\exp \begin{pmatrix} 0&x\\ -x^T&0\end{pmatrix} = \begin{pmatrix} I + \cos_x x x^T & \sin_x x\\ -\sin_xx^T & \cos(\norm x)\end{pmatrix},
\]
where we are using local coordinates around the south pole (whose coordinates are zero). One clearly has that $d \s(o): T_oS^2 \to \m$ is an isomorphism given by
\[
d \s(o) v = \begin{pmatrix} 0&v\\- v^T & 0\end{pmatrix}.
\]
The action of $\SO(3)$ on the sphere associated to this section, let's denote it by $g\cdot x$, is determined by the relation
\[
g \s(x) = \s(g\cdot x) h
\]
for some $h\in \SO(2)$ which is also determined by this relation. Let $g$ be as in (\ref{gson+1}). Straightforward calculations show that, if $\eta = g\cdot x$, then
\begin{equation}\label{soact1}
\sin_\eta \eta =  \sin_x\Theta x+\left(\cos_y\sin_x y^T x + \sin_y\cos(\norm{x})\right)\Theta y
\end{equation}
and
\begin{equation}\label{soact2}
\cos(\norm{\eta}) = \cos(\norm{y})\cos(\norm{x}) - \sin_x\sin_y y^T x.
\end{equation}
The last equation can be expressed in terms of the cosine of a certain angle using the spherical law of cosines, as we will see later. To have a better idea of what the coordinates given by this section are, recall the standard geometric identification of $\SO(3)/\SO(2)$ with the sphere: 

Given an element in $\SO(3)$ we can identify the last column with a point on the sphere and the first two columns as vectors tangent to the sphere at the point. The element of $\SO(2)$ acts on the two tangent vectors. With this identification, our coordinates result on 
\[
p = \begin{pmatrix} {\sin x}\frac x{\norm{x}}\\ \cos(\norm{x})\end{pmatrix} 
\]
being the point on the sphere. If we consider $\theta$ and $\phi$ to be the standard spherical angles, then $\norm{x} = \phi$ and $\frac{x}{\norm{x}} = \begin{pmatrix}\cos\theta\\ \sin\theta\end{pmatrix}$. Therefore, the coordinates $x$ describe the projection of $p$ on the $xy$ plane, multiplied by 
the angle $\phi$. See the picture below. This might seem as a cumbersome choice, but the advantages in calculations will be worthy, plus our Serret--Frenet equations will look very familiar to the reader.

\begin{figure}[tbh]
\centering
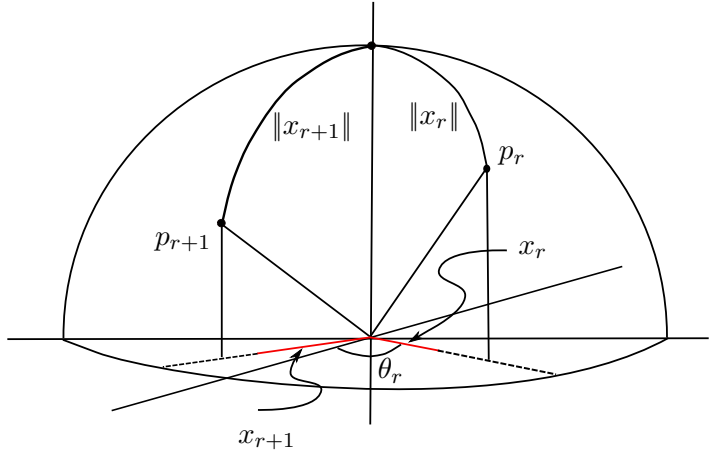\caption{Local coordinates on the sphere}
\end{figure}

 \subsubsection{Moving frame and invariants}
Next we will determine a moving frame using the normalization constants $c_r^r = 0$ and $c_{r+1}^r = a_r e_1$, where $a_r$ is an invariant still to be identified. Assume $\rho_r = g(\Theta_r, y_r)$ as in (\ref{gson+1}). The equation
\[
\rho_r \cdot x_r = 0
\]
determines the choice $y_r = -x_r$, and $\rho_r\cdot x_{r+1}$ is determined by the equation $\rho_r\cdot x_{r+1} = a_r e_1$ implicitly written as
\begin{equation}\label{ar}
\sin_{a_r e_1} a_r e_1 = \Theta_r\left(\sin_{x_{r+1}}x_{r+1}+(\cos_{x_r}\sin_{x_{r+1}}x_r^Tx_{r+1}-\sin_{x_r}\cos(\norm{x_{r+1}}))x_r\right).
\end{equation}
The invariant $a_r$ is determined by the condition $\Theta_r \in \SO(2)$, while this equation determines $\Theta_r$. If we impose the further condition $a_r>0$ for all $r$, then $\Theta_r$ is uniquely determined (it might belong to different connected components of $\SO(2)$ for different $r$'s). 

\begin{theorem} The right Maurer--Cartan matrices associated to $\rho_r$ as above are given by
\begin{equation}\label{sphereKr}
K_r = \exp\begin{pmatrix}0&k_r&0\\ -k_r&0&0\\0&0&0\end{pmatrix}\exp\begin{pmatrix} 0&0&-a_r\\ 0&0&0\\ a_r&0&0\end{pmatrix}.
\end{equation}
We call $a_r$ the discrete spherical arc-length invariant, and $k_r$ the discrete spherical curvature of the polygon.
\end{theorem}
\begin{proof}
Let us call $K_r = g(\Upsilon, \kappa)$ as in (\ref{gson+1}). To prove the theorem we can use the recursion equations (\ref{recursion}) for $K_r$ to determine it. For a right matrix the equations are given by $K_r \cdot c_{r+1}^r = c_{r+1}^{r+1}$, which, when substituted in (\ref{soact1}), and after minor simplifications, becomes
\[
0= \sin_{a_r} a_r e_1 + \left(\cos_{\kappa}\sin_{a_r}a_r \kappa^Te_1 + \sin_{\kappa}\cos{a_r}\right) \kappa.
\]
Let us re-write this as $0= \sin_{a_r} a_r e_1 + X_r \kappa$ so that $\kappa = -\frac{\sin_{a_r} a_r}{X_r} e_1$. We can in fact use the expression for $X_r$ to solve for it. Indeed, again after minor simplifications we obtain
\[
X_r = \cos_{\kappa}\sin_{a_r}a_r \kappa^Te_1 + \sin_{\kappa}\cos{a_r} = X_r\left(1-\cos\left(\frac{\sin a_r}{X_r}\right) + \frac{\cos a_r}{\sin a_r}\sin\left(\frac{\sin a_r}{X_r}\right)\right).
\]
This results in the relation
\[
\sin\left(\frac{\sin a_r}{X_r}\right)\cos a_r - \cos\left(\frac{\sin a_r}{X_r}\right)\sin a_r = \sin\left(\frac{\sin a_r}{X_r} - a_r\right) = 0
\]
which results in $X_r = \sin_{a_r}$ and $\kappa = -a_r e_1$.

This calculation determines the $\m$ component of $K_r$ as in the statement of the theorem. Notice that the $\h$ component $\Upsilon$ depends on only one parameter and therefore we can write it as in the statement, even if we do not give its explicit formula.
\end{proof}

Both invariants $a_r$ and $k_r$ have a very simple geometric description, as shown in our next theorem.
\begin{theorem} The spherical arc-length invariant $a_r$ is the length of the arc joining $x_r$ to $x_{r+1}$. Let $\beta_{r,s}$ be the counterclockwise angle formed by the arc $x_rx_{s}$ and the arc $x_s N$, where $N$ is the north pole. Then,  
\[
k_{r-1} = \pi -(\beta_{r,r+1}+\beta_{r+1,r+2})
\]
\end{theorem}

\begin{figure}[tbh]
\centering
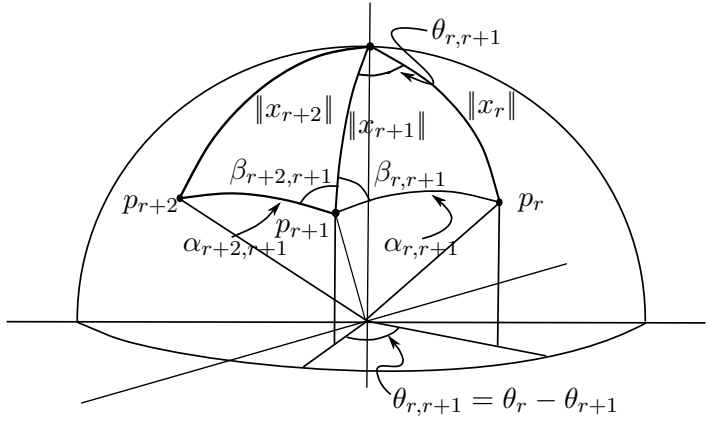\caption{Description of the discrete curvature. Notice that when $\beta_{r+1,r+2} + \beta_{r,r+1} = \pi$ the curvature is zero.}
\end{figure}

\begin{proof}
Calculating the length of the vector in equation (\ref{ar}) we see that 
\[
\sin^2(a_r)= \sinro^2+\left((\cosr-1)\sinro\cos\theta_{r,r+1}-\sinr\cosro\right)^2
\]
where $\theta_{r,r+1} = \theta_r-\theta_{r+1}$ and $\theta_r$ is the polar angle of $x_r$. After some simplifications this equality becomes
\[
\sin^2(a_r) = 1-\left(\cosr\cosro+\sinr\sinro\cos\theta_{r,r+1}\right)^2.
\]
 Using the spherical law of cosines (see figure 3) we obtain that 
 \[
 \cosr\cosro+\sinr\sinro\cos(\theta_{r,r+1}) = \cos \alpha_{r,r+1},
 \]
  where $\alpha_{r,r+1}$ is the angle between $p_r$ and $p_{r+1}$, the points on the sphere corresponding to $x_r$ and $x_{r+1}$. Using this fact, the above becomes
 \[
 -\cos^2\alpha_{r,r+1} + 1 = \sin^2\alpha_{r,r+1} = \sin^2a_r.
 \]
 Therefore, $a_r = |\alpha_{r,r+1}|$ since it is a positive number, as stated. 
 
 \begin{figure}[tbh]
\centering
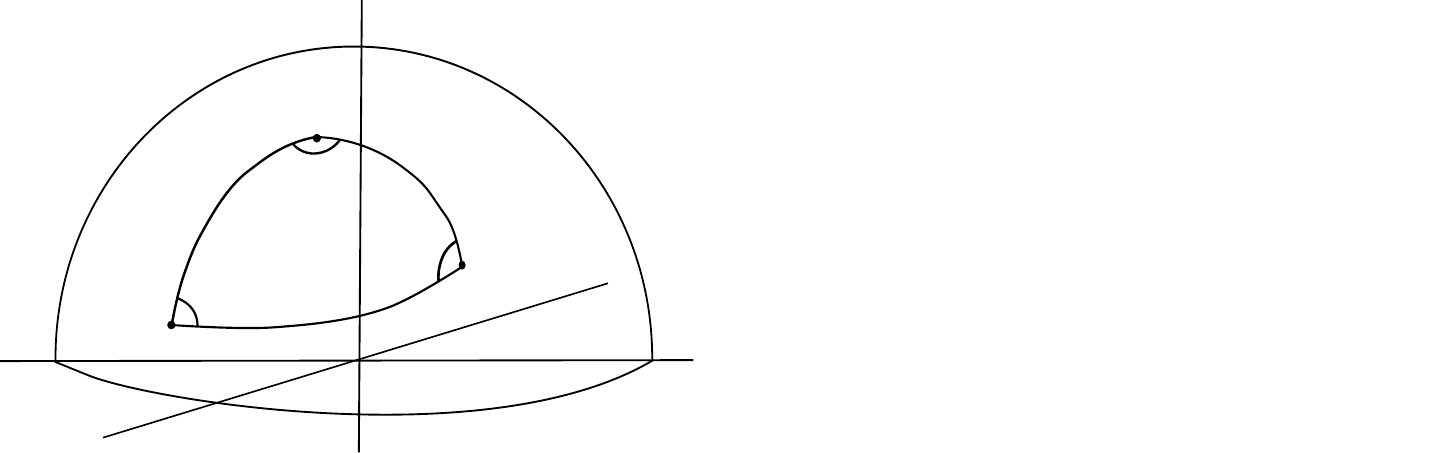\caption{Spherical trigonometric law of cosines and sines. The angles $A$, $B$ and $C$ are the corner angles while $a,b,c$ are the angles corresponding to the indicated arc as measured from the center of the sphere.}
\end{figure}

A similar calculation, although a more involved one, describes $k_r$. Using the recursion relation
\[
K_r\cdot c_{r+1}^{r+1} = c_{r+1}^{r+2},
\] 
the fact that $c_{r+1}^{r+1} = 0$, the expression of $K_r$ and the formula (\ref{soact1}) for the action, we have that
\[
\sin_{c_{r+1}^{r+2}}c_{r+1}^{r+2} = \sin a_r\begin{pmatrix}-\cos k_r\\ \sin k_r\end{pmatrix}.
\]
Therefore, {\it $\pi-k_r$ is the polar angle of $c_{r+1}^{r+2}$} (while $\norm{c_{r+1}^{r+2}} = a_r$). Now, we do know that $\rho_{r+2}\cdot x_{r+1} = c_{r+1}^{r+2}$, where $\rho$ is the right moving frame; therefore, we can again use (\ref{soact1}) to find
\begin{equation}\label{veq1}
\sin_{c_{r+1}^{r+2}}c_{r+1}^{r+2} = \Theta_{r+2} V_{r+1,r+2}
\end{equation}
where
\[
V_{i,j} = \sin_{x_{i}}x_{i} + \left(\cos_{x_{j}}\sin_{x_{i}} x_{j}\cdot x_{i} - \sin_{x_{j}}\cos\norm{x_i}\right)x_{j}.
\]
We clearly see that the polar angle of $c_{r+1}^{r+2}$ is a combination of the polar angle of $\Theta_{r+2}$ and the one of $V_{r+1,r+2}$. Now, we know that $\Theta_r$ is determined by the equation (\ref{ar}). That is, $\Theta_{r+2}$ rotates $V_{r+3,r+2}$ to the $x$-axis
\[
\Theta_{r+2}\left( \sin_{x_{r+3}}x_{r+3} + \left(\cos_{x_{r+2}}\sin_{x_{r+3}} x_{r+2}\cdot x_{r+3} - \sin_{x_{r+2}}\cos\norm{x_{r+3}}\right)x_{r+2}\right) 
\]
\[
= \Theta_{r+2}V_{r+3,r+2} =\sin_{a_{r+2}e_1}a_{r+1}e_1 .
\]
If we denote by $P(v)$ the polar angle of the vector $v$, from here we get that $\pi-k_r = P(V_{r+1,r+2}) - P(V_{r+3,r+2})$.

Next, we can write $x_r$ in terms of the polar coordinates of $p_r$, that is, 
\[
x_r = \begin{pmatrix} \sinr \cos\theta_r\\ \sinr\sin\theta_r\end{pmatrix}
\]
and substitute in the expression for $V_{i,j}$. After some minor trigonometric manipulations we obtain
\[
V_{i,j} = \begin{pmatrix} \cos \theta_j&-\sin\theta_j\\ \sin\theta_j&\cos\theta_j\end{pmatrix}\begin{pmatrix}\cos\norm{x_j}\sin\norm{x_i}\cos(\theta_{i,j})-\sin\norm{x_j}\cos\norm{x_i}\\ \sin\norm{x_i}\sin(\theta_{i,j})\end{pmatrix}
\]
Form here $P(V_{r+1,r+2}) - P(V_{r+3,r+2}) = P(v_{r+1,r+2}) - P(v_{r+3,r+2})$, where
\begin{equation}\label{vij}
v_{i,j} = \begin{pmatrix}\cos\norm{x_j}\sin\norm{x_i}\cos(\theta_{i,j})-\sin\norm{x_j}\cos\norm{x_i}\\ \sin\norm{x_i}\sin(\theta_{i,j})\end{pmatrix}.
\end{equation}
Notice that, as before, after basic trigonometric manipulations
\[
\norm{v_{i,j}}^2 =  1-\left(\cos(\norm{x_i})\cos(\norm{x_j})+\sin(\norm{x_i})\sin(\norm{x_j})\cos(\theta_{i,j})\right)^2 = 1-\cos^2\alpha_{i,j}= \sin^2\alpha_{i,j}.
\]
where we have used the spherical law of cosines once more. Notice that $0\le \alpha_{i,j}\le \pi$ and so from (\ref{vij})  the sine of the polar angle of $v_{i,j}$ is given by
\[
\frac{\sin\norm{x_i}\sin(\theta_{i,j})}{\sin\alpha_{i,j}}.
\]
Let us assume all the angles involved are in the first two quadrants. Then, the spherical law of sines tells us that
\[
\frac{\sin\alpha_{r+1,r+2}}{\sin\theta_{r+1,r+2}} = \frac{\sinr}{\sin\beta_{r+1,r+2}}
\]
and so the sine of the polar angle of $v_{r+1,r+2}$ is given by 
\[
\frac{\sin\norm{x_{r+1}}\sin(\theta_{r+1,r+2})}{\sin\alpha_{r+1,r+2}} = \sin\beta_{r+1,r+2}.
\]
Likewise, the polar angle of $v_{r+3,r+2}$ is $\beta_{r+3,r+2} = -\beta_{r+2,r+3}$. Therefore, 
\[
\pi - k_r = \beta_{r+1,r+2} +\beta_{r+2,r+3}
\]
as stated.
\end{proof}
 \subsubsection{Invariantization of invariant evolutions and its associated integrable system}
We will next describe the evolution induced on $K_r$ by an invariant evolution of polygons on the sphere. According to (\ref{invev}) and using our definition of moving frame, the most general form for an invariant evolution of polygons on the sphere is given by
\[
(x_s)_ t = \left( \sin_{x_s}^{-1} \left(I - \frac{x_s x_s^T}{\norm {x_s}^2}\right) + \frac{x_sx_s^T}{\norm {x_s}^2}\right) \Theta_s^{-1} \r_s
\]
for some invariant vector $\r_s = \begin{pmatrix} r_s^1\\ r_s^2\end{pmatrix}$ depending on $(k_r)$ and $(a_r)$.

\begin{theorem} Assume the twisted polygons $(x_r)$ are solutions of an invariant evolution of the form above. Then the invariants $a_s$ and $k_s$ evolve following the equations
\begin{eqnarray}\label{aev}
(a_s)_t &=& -r_s^1+r_{s+1}^1\cos k_s-r_{s+1}^2\sin k_s\\
\label{kev}
(k_s)_t &=& r_s^2\frac1 {\sin a_s}-r_{s+1}^2\frac{\cos a_{s+1}}{\sin a_{s+1}}+r_{s+2}^1 \frac{\sin k_{s+1}}{\sin a_{s+1}} + r_{s+2}^2\frac{\cos k_{s+1}}{\sin a_{s+1}} \\ &-&   r_{s+1}^1\frac{\cos a_s}{\sin a_s} \sin k_s- r_{s+1}^2\frac{\cos a_s}{\sin a_s}\cos k_s 
\end{eqnarray}
\end{theorem}
\begin{proof}
Assume the polygons are evolving following this evolution; then, our {\it left} Maurer--Cartan matrices $\widehat K_s  = K_s^{-1}$ will follow the evolution
\begin{equation}\label{structsphere}
\widehat K_s^{-1} (\widehat K_s)_t = N_{s+1} - \widehat K_s^{-1} N_s \widehat K_s
\end{equation}
where 
\[
N_s = (\rho_s)_t\rho_s^{-1} = \begin{pmatrix} 0&n_s & r_s^1\\ -n_s & 0 & r_s^2\\ -r_s^1 &- r_s^2 & 0\end{pmatrix}.
\]

Next, we remark that $\begin{pmatrix} 0&-\kappa\\ \kappa^T & 0\end{pmatrix}$ commutes with its derivative since $\kappa$ is a multiple of $e_1$. This means 
\[
\left(\exp\begin{pmatrix} 0&-\kappa\\ \kappa^T & 0\end{pmatrix}\right)_t = \begin{pmatrix}0&0&(a_s)_t\\ 0&0&0\\ -(a_s)_t & 0&0\end{pmatrix}\exp\begin{pmatrix} 0&-\kappa\\ \kappa^T & 0\end{pmatrix}.
\]
Using the expression of $K_r$ in the previous theorem, one has that
\[
\widehat K_s^{-1}(\widehat K_s)_t  = \begin{pmatrix}0&-(k_s)_t&(a_s)_t\cos k_s\\ (k_s)_t&0&-(a_s)_t\sin k_s\\ -(a_s)_t\cos k_s&(a_s)_t\sin k_s&0\end{pmatrix}.
\]
Substituting this into (\ref{structsphere}) and multiplying out matrices we get
\[
 (k_s)_t = - n_{s+1} + n_s \cos a_s + r_s^2 \sin a_s
 \]
 and
 \[
 \begin{pmatrix} \cos k_s&\sin k_s\\-\sin k_s&\cos k_s\end{pmatrix} \begin{pmatrix} (a_s)_t\\ 0\end{pmatrix} =  \begin{pmatrix} \cos k_s&\sin k_s\\-\sin k_s&\cos k_s\end{pmatrix}\begin{pmatrix} - r_s^1\\-r_s^2\cos a_s+n_s\sin a_s\end{pmatrix} + \begin{pmatrix} r_{s+1}^1\\ r_{s+1}^2\end{pmatrix}.
 \]
 From here we get
 \[
 n_s = r_s^2\frac{\cos a_s}{\sin a_s}-\frac 1{\sin a_s}\left(r_{s+1}^1\sin k_s+ r_{s+1}^2\cos k_s\right)
 \]
and
\[
(a_s)_t = -r_s^1- r_{s+1}^2\sin k_s+r_{s+1}^1\cos k_s.
\]
Substituting the vale of $n_s$ in the evolution of $k_s$ completes the proof.
\end{proof}

In particular, if we ask that the evolution preserves the spherical arc-length (i.e. $(a_s)_t = 0$) and we choose the value $\tan a_s = 1$ as constant value, then choosing $r_s^1 = 1$ for all $s$ produces the curvature evolution
\[
(k_s)_t = \sqrt{2}\left(\frac{1-\cos k_{s+1}}{\sin k_{s+1}} - \frac{1-\cos k_{s-1}}{\sin k_{s-1}}\right). 
\]

Let $u_s=\frac{1-\cos k_{s}}{\sin k_{s}}$. The above equation becomes
\begin{eqnarray}\label{volm}
 (u_s)_t=\frac{\sqrt{2}}{2} (1+u_s^2) (u_{s+1}-u_{s-1}), 
\end{eqnarray}
which is a special case of equation (V1) in the list of integrable Volterra-type equations \cite{Yami}.
It is a bi-Hamiltonian equation \cite{mr93c:58096}, where the second Hamiltonian operator was not explicitly given. 
Here we give its compatible Hamiltonian and symplectic operators \cite{mr94j:58081}.
Equation (\ref{volm}) can be written as
\[(u_s)_t=\frac{\sqrt{2}}{4} \HH \delta_{u_s} \ln (1+u_s^2)\quad \mbox{and}\quad 
{\mathcal I} (u_s)_t=\frac{\sqrt{2}}{2}\delta_{u_s} ((1+u_s^2) u_{s+1} u_{s-1}+\frac{1}{2} u_s^2 u_{s-1}^2)
 \]
where
\begin{eqnarray*}
&&\HH=(1+u_s^2) \left( \T - \T^{-1} \right)(1+u_s^2) 
\end{eqnarray*}
is Hamiltonian and 
\begin{eqnarray*}
&&{\mathcal I}=\T-\T^{-1}+\frac{2u_s}{1+u_s^2}(\T-1)^{-1}(u_{s+1}+u_{s-1})
+(u_{s+1}+u_{s-1})\T(\T-1)^{-1}\frac{2u_s}{1+u_s^2}
\end{eqnarray*}
is a symplectic operator. Furthermore, these two operators satisfy
\begin{eqnarray*}
 \HH {\mathcal I}=\Re^2=\left((1+u_s^2) \T+2 u_s u_{s+1} +(1+u_s^2)\T^{-1}
+(u_s)_t (\T-1)^{-1} \frac{2u_s}{1+u_s^2}\right)^2,
\end{eqnarray*}
where $\Re$ is a Nijenhuis recursion operator of (\ref{volm}). 
Its recursion operator $\Re$ can not be written as
the product of weakly nonlocal Hamiltonian and symplectic operators. A similar example
was presented in \cite{wang09}.

Notice that, as before, $\T-1$ is not invertible in the periodic case. That means we will need to work on infinite gons, or assume we are working with Hamiltonians whose gradient is in the image of the operator, as it is the case with this example.

\section{Conclusion and Future Work}

In this paper, we have developed a notion of a discrete moving frame and shown it has computational advantages over a single frame for the invariantization of discrete evolution flows and mappings. 
Further, we have shown  in our examples  that the use of discrete moving frames greatly aids the identification of discrete integrable systems and biPoisson maps that can be obtained as invariantizations of invariant 
flows of gons. 

Investigations are under way on how the Hamiltonian structures might appear in the general case, and in particular how they appear in the projective plane. 
Further work also remains to illustrate and detail how our constructions of invariant mappings and their invariantizations lead in general to discrete integrable mappings, as seen in Section \ref{invmapaswell}. 
This work is different in nature to the previous one and it will be greatly aided by understanding the differential-difference problem. 

Hamiltonian structures  in the projective plane could be relevant to the results in \cite{OST}. There the authors proved that the so-called {\it pentagram} map is completely integrable, but they did not provide a biHamiltonian structure for its invariantization. A structure obtained in this setting would be a natural candidate to prove that the pentagram map is biPoisson.

The study of discrete moving frames for more general applications, and insight into how our theorems in Section 3 here may generalize, remains to be achieved. One such application is to the discrete Calculus of Variations, as briefly indicated
in Example \ref{disvarScalTransDisFrame}. We believe however, that eventually the applications for a more general theory of discrete moving frames will extend to invariant numerical schemes and computer graphics, amongst others.


\begin{thebibliography}{99}

\bibitem{A} Adler  M.,  (1979). { On a Trace Functional for Formal Pseudo-differential Operators and the
Symplectic Structure of the KdV}, {\frenchspacing\em  Invent. Math.} {\bf 50}, 219--248.


\bibitem{boutin} 
Boutin M.\ (2002) 
On orbit dimensions under a simultaneous Lie group
action on n copies of a manifold, {\it J.\ Lie Theory\/}, {\bf 12}, 191--203.

\bibitem{Cartan}
Cartan E.\, (1952--55) Oeuvres compl\`etes,
  Gauthier-Villiars.

\bibitem{chhay}
Chhay, M.\ and Hamdouni, A.\ (2010) A new construction for invariant numerical schemes using moving frames,
{\frenchspacing\em  C. R. Acad. Sci. Meca.} {\bf 338}, 97--101.
\bibitem{mr94j:58081}
Irene Dorfman.
\newblock {\em Dirac structures and integrability of nonlinear evolution
  equations}.
\newblock John Wiley \& Sons Ltd., Chichester, 1993.



\bibitem{refFelOla}
Fels, M.\ and  Olver, P.J.\ (1998)
Moving Coframes I,
{\frenchspacing\em Acta Appl. Math.\/} {\bf 51} 161--213.

\bibitem{refFelOlb}
Fels, M.\ and  Olver, P.J.\ (1999), Moving Coframes II,
{\frenchspacing\em Acta Appl. Math.\/} {\bf 55}  127--208.


\bibitem{Flaschka1}
Flaschka, H. (1974), The {T}oda lattice. II. Existence of integrals,
{\frenchspacing\em Phys. Rev. B\/} {\bf 9}(4) 1924--1925.

\bibitem{Flaschka2}
Flaschka, H. (1974), On the {T}oda Lattice. II Inverse-Scattering Solution,
{\frenchspacing\em Progress of Theoretical Physics\/} {\bf 51}(3) 703--716.


\bibitem{GonMan} Gon\c{c}alves, T.M.N.\ and Mansfield E.L., (2011)
On Moving frames and Noether's Conservation Laws. { \frenchspacing\em Studies in Applied Mathematics\/} {\bf 128} 1--29
doi 10.1111/j.1467-9590.2011.00522.x
 
\bibitem{Green}
Green, M.L.\ (1978) The moving frame, differential invariants and rigidity
theorems for curves in homogeneous spaces, {\frenchspacing\em
Duke Math. J.\/} {\bf 45}, 735--779.

\bibitem{Griffiths}
Griffiths, P.\ (1974) On Cartan's methods of Lie groups and moving frames
as applied to uniqueness and existence questions in differential
geometry, {\frenchspacing\em Duke Math. J.\/} {\bf 41}, 775--814.

\bibitem{HLM}  Heredero, R.\   Lopez, A.\  and Mar\'\i~Beffa,  G.\  (1997)
 Invariant Differential Equations and the Adler-Gel'fand-Dikii bracket ,  {\frenchspacing\em J. Math. Phys.} {\bf 38}, 5720--5738. 

\bibitem{HerHick}
Hickman, M.S.\ and Hereman, W.\ (2003) Computation of densities and fluxes of nonlinear differential-difference equations,
{\frenchspacing\em Proc. Roy. Soc. A} {\bf 459}, 2705--2729.

\bibitem{hubertAA}
Hubert, E.\ (2005),
Differential Algebra for Derivations with Nontrivial Commutation Rules,
{\em J.\ of Pure and Applied Algebra\/}, {\bf 200} (1-2), 163--190.

\bibitem{hubertAC}
Hubert, E.\ (2009a),
Differential invariants of a Lie group action: syzygies on a generating set.
{\it J.\ Symbolic Computation\/} {\bf 44}(4), 382--416.

\bibitem{hubertAD}
Hubert, E. \ (2009b),
Generation properties of Maurer--Cartan invariants.
Preprint [hal:inria-00194528/en]



\bibitem{hubertA}
Hubert, E.\ and Kogan I.A.\ (2007a),
Smooth and Algebraic Invariants of a Group Action. Local and Global Constructions.
{\it Foundations of Computational Mathematics\/}, {\bf 7} (4), 345--383.

\bibitem{hubertB}
Hubert, E.\ and Kogan I.A.\ (2007b),
Rational Invariants of a Group Action. Construction and Rewriting.
{\it Journal of Symbolic Computation\/}, {\bf 42} (1-2), 203--217.

\bibitem{HyMan}
Hydon, P.E.\ and Mansfield E.L.\  (2004) A variational complex for difference equations, {\em Foundations of
Computational Mathematics} {\bf 4}, 187--217.

\bibitem{kacm}
Kac, M.\ and van Moerbecke, P.\ (1975), On an explicitly soluble system of 
nonlinear differential equations related to certain Toda lattice. 
{\frenchspacing\em Adv. Math.\/} {\bf 16},  160--169.
\bibitem{kimA}
Kim, P.\ and Olver, P.J.\ (2004) Geometric integration via multi-space, {\em Regular and Chaotic Dynamics\/}, 9(3), 213--226.

\bibitem{kimB}
Kim, P.\ (2007) Invariantization of Numerical Schemes Using Moving Frames, {\em BIT Numerical Mathematics\/}  47(3), p.525. 

\bibitem{kimC}
Kim, P.\ (2008) Invariantization of the Crank-Nicolson Method for Burgers' Equation,  {\em Physica D: Nonlinear Phenomena\/},
 {\bf 237}(2), p.243.
 
 \bibitem{kogan}
Kogan, I.A.\ and Olver, P.J.\ (2003),
Invariant Euler-Lagrange equations and the invariant variational bicomplex
{\frenchspacing\em Acta Appl. Math.\/} {\bf 76},  137--193.

\bibitem{kp85}
Kupershmidt, B.A. (1985)
\newblock {\em Discrete Lax equations and differential-difference calculus}.
\newblock Asterisque.

 
   \bibitem{LN}  Lobb S.B.\ and Nijhoff F.W.\  (2010) {\em Lagrangian multiform structure for the lattice Gel'fand-Dikii hierarchy}, J. Phys. A {\bf 43} (7)
   072003.

\bibitem{manakov}
Manakov, S.V.\ (1975),
Complete integrability and stochastization in discrete dynamical systems.
{\frenchspacing\em Sov.Phys. JETP\/} {\bf 40},  269--274.

\bibitem{mansfield}
Mansfield, E.L.\ (2010), A practical guide to the invariant calculus, Cambridge
University Press, Cambridge.

\bibitem{ManHy} 
Mansfield, E.L.\ and Hydon, P.E., (2001) Towards approximations of difference equations that preserve integrals, Proc.\ 2001 International
Symposium on Symbolic and Algebraic Manipulation (ISSAC 2001) ed., B.\ Mourrain, ACM, New York,  217--222

\bibitem{ManKamp}
Mansfield, E.L.\ and van der Kamp, P.\  (2006) Evolution of curvature invariants and lifting integrability, {\em J.\ Geometry and Physics\/},
{\bf 56} 1294--1325.


\bibitem{M1}  Mar\'\i~Beffa, G.\ (2008)  {\em Geometric Hamiltonian structures on flat semisimple homogeneous manifolds}, The Asian Journal of Mathematics, {\bf 12}(1), 1--33. 
\bibitem{M2}  Mar\'\i~Beffa,  G.\  (2006)  {\em Poisson geometry of differential invariants of curves in some nonsemisimple homogenous spaces}, Proc. Amer. Math. Soc. {\bf 134}, 779--791.
\bibitem{M3}  Mar\'\i~Beffa, G.\  (2010) {\em On bi-Hamiltonian flows and their realizations as curves in real semisimple homogeneous manifolds}, Pacific Journal of Mathematics, {\bf 247-1},  163--188.
\bibitem{M4} Mar\'\i~Beffa, G.\  (1999)  {\em The Theory of differential invariants and KdV Hamiltonian evolutions}, Bull. Soc. math. France, {\bf 127}  363--391.  





\bibitem{OJS}
Olver, P.J.\  (2001) Joint invariant signatures, Found. Comput. Math. {\bf 1},  3--67.


\bibitem{Ogen}
 Olver, P.J.\ (2001) Moving frames -- in geometry, algebra, computer vision, and numerical
analysis, in: Foundations of Computational Mathematics, R. DeVore, A. Iserles and E. Suli, eds., London Math. Soc. Lecture Note Series, vol. 284, Cambridge University Press, Cambridge, 267--297. 



\bibitem{OST}  Ovsienko, V.\, Schwartz, R.\ and Tabachnikov S.\  (2010) 
{\em The Pentagram map: a discrete integrable system}, Communications in Mathematical Physics,  {\bf 299}(2),  409--44. 

\bibitem{suris03}
Suris, Y.B.
\newblock {\em The Problem of Integrable Discretization: {H}amiltonian
  Approach}.
\newblock Progress in Mathematics, Vol. 219. Birkh{\"a}user, Basel, 2003.
\bibitem{toda}
Toda, M.
\newblock {\em Theory of nonlinear lattice}.
\newblock Springer Series in Solid-State Sciences, 20, Springer-Verlag, Berlin, 1989.
\bibitem{wang09}
 Wang, J.P. (2009)
\newblock Lenard scheme for two-dimensional periodic volterra chain.
\newblock {\em J. Math. Phys.}, {\bf 50}:023506.

\bibitem{Yami} Yamilov, R.I.  {\em Symmetries as integrability criteria for differential difference equations},
Journal of Physics A: Mathematical and General, (2006) Vol {\bf 39}, R541--R623.
\bibitem{mr93c:58096}
Zhang, H., Tu, G.-Z., Oevel, W. and Fuchssteiner, B. (1991)
\newblock Symmetries, conserved quantities, and hierarchies for some lattice
  systems with soliton structure.
\newblock {\em J. Math. Phys.}, {\bf 32}(7), 1908--1918.
\end{thebibliography}
\end{document}